\documentclass{article}
\usepackage{amsthm,amsmath,amsfonts}
\usepackage{algorithmic}
\usepackage{algorithm}
\usepackage{array}
\usepackage[caption=false,font=normalsize,labelfont=sf,textfont=sf]{subfig}
\usepackage{textcomp}
\usepackage{stfloats}
\usepackage{url}
\usepackage{verbatim}
\usepackage{graphicx}
\usepackage{cite}
\usepackage{newtxtext}
\usepackage{amssymb, amsmath,threeparttable, bbding, url,xcolor}
\usepackage{booktabs}
\usepackage{circuitikz}
\usepackage{subcaption}
\newtheorem{theorem}{Theorem}[section]
\newtheorem{lemma}[theorem]{Lemma}
\newtheorem{proposition}[theorem]{Proposition}
\newtheorem{example}[theorem]{Example}
\newtheorem{remark}[theorem]{Remark}
\newtheorem{corollary}[theorem]{Corollary}

\ctikzset{logic ports=ieee}
\newcommand{\cchi}{\ooalign{%
  $\chi$\cr
  \hidewidth$\,\,\chi$\hidewidth
}%
\,}
\DeclareRobustCommand{\rcchi}{\cchi}

\newcommand{\F}{{\mathbb{F}}}
\newcommand{\1}{\mathbf{1}}
\newcommand{\ord}{{\mathrm{ord}}}

\begin{document}

\title{A Generalized $\chi_n$-Function}

\author{Cheng~Lyu, Mu~Yuan, Dabin~Zheng, Siwei Sun 
	and  Shun Li
\thanks{Cheng~Lyu, Mu Yuan, and Dabin Zheng are with Hubei Key Laboratory of Applied Mathematics, Faculty of Mathematics and Statistics, Hubei University. Dabin Zheng is also with Key Laboratory of Intelligent Sensing System and Security (Hubei University), Ministry of Education, Wuhan, 430062, China. E-mail: chenglyu@139.com; yuanmu847566@outlook.com; dzheng@hubu.edu.cn. The corresponding author is Dabin Zheng.}
\thanks{Siwei Sun and Shun Li are with School of Cryptology, University of Chinese Academy of Sciences, Beijing 100049, China. E-mail: sunsiwei@ucas.ac.cn; lishun@ucas.ac.cn.}    
\thanks{The research was supported in part by the China National Key Research and Development Program under Grant 2021YFA1000600, and
 in part by the Key Project of the National Natural Science Foundation of China under Grant 62032014, and in part by the Nature Science Foundation of China (NSFC) under  Grants 62272148, 12301671, and in part by the Key Project of the National Cryptography Science Foundation under Grant 2025NCSF01012, and in part by China Postdoctoral Science Foundation 2024M763221.}}




\maketitle

\begin{abstract}
The mapping $\chi_n$ from $\F_{2}^{n}$ to itself defined by $y=\chi_n(x)$ with $y_i=x_i+x_{i+2}(1+x_{i+1})$, where the indices are computed modulo $n$, has been widely studied for its applications in lightweight cryptography. However, $\chi_n $ is bijective on $\F_2^n$ only when $n$ is odd, restricting its use to odd-dimensional vector spaces over $\F_2$. To address this limitation, we introduce and analyze the generalized mapping $\chi_{n, m}$ defined by $y=\chi_{n,m}(x)$ with $y_i=x_i+x_{i+m} (x_{i+m-1}+1)(x_{i+m-2}+1) \cdots (x_{i+1}+1)$, where $m$ is a fixed integer with $m\nmid n$. To investigate such mappings, we further generalize $\chi_{n,m}$ to $\theta_{m, k}$, where $\theta_{m, k}$ is given by $y_i=x_{i+mk} \prod_{\substack{j=1,\,\, m \nmid j}}^{mk-1}  \left(x_{i+j}+1\right), \,\,{\rm for }\,\, i\in \{0,1,\ldots,n-1\}$. We prove that these mappings generate an abelian group isomorphic to the group of units in $\F_2[z]/(z^{\lfloor n/m\rfloor +1})$. This structural insight enables us to construct a broad class of permutations over $\F_2^n$ for any positive integer $n$, along with their inverses. We rigorously analyze algebraic properties of these mappings, including their iterations, fixed points, and cycle structures. Additionally, we provide a comprehensive database of the cryptographic properties for iterates of $\chi_{n,m}$ for small values of $n$ and $m$. Finally, we conduct a comparative security and implementation cost analysis among $\chi_{n,m}$, $\chi_n$,  $\cchi_n$ (EUROCRYPT 2025 \cite{belkheyar2025chi}) and their variants, and prove Conjecture~1 proposed in~\cite{belkheyar2025chi} as a by-product of our study. Our results lead to generalizations of $\chi_n$, providing alternatives to $\chi_n$ and $\rcchi_n$. 
\end{abstract}


\section{Introduction}
In order to find a shift-invariant transformation over $\F_2^n$ that is easy to implement and has good cryptographic properties, Joan Daemen in~\cite{daemen1995cipher} introduced a mapping $\chi_n\colon \F_2^n \to \F_2^n$, $ x\mapsto y$ defined by $y_i = x_i +(x_{i+1} +1) x_{i+2}$, where the indices are computed modulo $n$. Due to its compactness and favorable cryptographic properties, $\chi_n$ has been used in several prominent cryptographic algorithms, including \texttt{KECCAK}-$f$~\cite{bertoni2008keccak} (part of the NIST standard SHA-3~\cite{nist2015sha}) and Ascon~\cite{dobraunig2021ascon} (the NIST standard for lightweight cryptography~\cite{nist2023ascon}).

Recent years have witnessed significant progress in the study of the mapping $\chi_n$, with notable advancements in both cryptanalytic and algebraic domains. From a cryptanalytic perspective, Daemen et al. in~\cite{daemen2021comput} and Mella et al. in~\cite{mella2023diff} conducted comprehensive investigations into the differential and linear propagation properties of 
$\chi_n$, respectively. On the algebraic front, Biryukov et al. in~\cite{biryukov2014asasa} pioneered the analysis of the algebraic degree of the inverse mapping $\chi_n^{-1}$ of $\chi_n$, demonstrating its resilience against meet-in-the-middle attacks. This foundational work was further extended by Liu et  al. in \cite{liu2022inverse}, who derived an explicit algebraic expression for $\chi_n^{-1}$. Schoone and Daemen in~\cite{schoone2024state} explored the order and cycle structure of $\chi_n$, while their subsequent research~\cite{schoone2024algebraic} investigated it as a polynomial mapping, leading to the conjecture that  $\chi_n$ does not act as a bijection of a vector space over any extension field of $\F_2$. This conjecture was later rigorously proven by Graner et al. in~\cite{graner2024bijectivity}, who established that the mapping $\chi_n\colon \F_q^n \rightarrow\F_q^n$ is a permutation if and only if $q = 2$ and $ n$ is odd. These contributions collectively deepen the understanding of $\chi_n$ and its cryptographic implications.

Very recently, Kriepke and Kyureghyan in~\cite{kriepke2024algebraic} provided an in-depth understanding of the algebraic properties of $\chi_n$. Specifically, they observed that all iterates of $\chi_n$ can be expressed as linear combinations of the mappings $\gamma_{2k}$ for $k\geq 1$ , where the $i$-th coordinate of the output is given by
	\begin{equation*}
		y_i=x_{i+2k} (x_{i+2k-1}+1) (x_{i+2k-3}+1) \cdots (x_{i+1}+1)
	\end{equation*} 
for $i\in \{0,1,\cdots,n-1\}.$
Based on this observation, they conducted a comprehensive study of the set $\Gamma$ spanned by these $\gamma_{2k}$, with a particular focus on  its subset $G$ composed of invertible elements with respect to composition operation. Finally, by showing that $G$ is isomorphic to an abelian group $(\F_2[z]/ (z^{(n+1)/2}))^*$, they derived concise expressions for iterates of $\chi_n$ and explained their properties.

Note that $\chi_n$ is a bijection on $\F_2^n$ if and only if $n$ is odd. To address this limitation and extend bijective capabilities to even dimensions, we introduce a generalized mapping class, denoted by $\chi_{n,m}$. This mapping  $\chi_{n,m}\colon \F_{2}^{n} \to \F_{2}^{n}$ is defined as  $y=\chi_{n,m}(x)$, where the $i$-th coordinate of $y$ is represented as
\begin{equation} \label{eq-chinm}
	y_i=x_i+x_{i+m} (x_{i+m-1}+1)(x_{i+m-2}+1) \cdots (x_{i+1}+1)
\end{equation}
for $m \ge 2$ and $i \in \{0,1,\ldots,n-1\}$. Crucially, $\chi_{n,m}$ acts as a bijection on $\F_2^n$ if and only if $m \nmid n$,
thereby enabling its application to both odd and even dimension vector spaces over $\F_2$. 
Building on the framework established in~\cite{kriepke2024algebraic}, we generalize  $\gamma_{2k}$ in~\cite{kriepke2024algebraic} to $\theta_{m,k}$. We will demonstrate that these mappings $\theta_{m,k}$ generate a linear space $\Theta_{n,m}$ over $\F_{2}$, and that the subset $G_{n,m}$ forms an abelian group under composition. This construction yields a broad class of bijective mappings of $\F_2^n$, including $\chi_{n,m}$. Leveraging this algebraic structure, we derive precise results on the {cycle structure}, algebraic degree, compositional inverse, and fixed points of the compositional iterates of $\chi_{n,m}$. Notably, $\chi_n$ emerges as a special case of 
$\chi_{n,m}$ when $m=2$, and our findings significantly extend previously known results. Furthermore, while $\chi_n$ is an involution on $\F_{2}^n$ only when $n \leq 3$,  involutory mappings $\chi_{n,m}$ exist even for larger~$n$.

Most recently, a variant permutation of $\chi_n$ on even-dimensional vector spaces $\F_2^{2k}$ has been introduced by Belkheyar et. al. in~\cite{belkheyar2025chi}, named C{\scriptsize HI}C{\scriptsize HI} or $\cchi_{2k}$, defined as follows:
\begin{equation} \label{eq-cchi}
	y_i=\left\{\begin{aligned} 
		&x_i+ \overline{x}_{i+1} x_{i+2}, & & i < k-3 \text{ or } k< i < 2k-2,\\
		&x_k+ \overline{x}_{k-2} x_0, & & i=k-3,\\
		&x_{k-1} + \overline{x}_0 x_1, && i= k-2,\\
		&\overline{x}_{k-3} + \overline{x}_k \, \overline{x}_{k+1}, & & i=k-1,\\
		&x_{k-2}+ \overline{x}_{k+1} x_{k+2},& & i=k,\\
		&x_{2k-2}+ \overline{x}_{2k-1} x_{k-1},& & i=2k-2,\\
		&x_{2k-1}+ \overline{x}_{k-1}x_k,&& i=2k-1,
	\end{aligned}  \right.
\end{equation}
where $\bar{x}_i\colon=x_i+1$ for $i \in \{0,1,\ldots, 2k-1\}$, and $k$ is even. Note that $\cchi_{2k}$ is extended affine (EA) equivalent to the  concatenation of $\chi_{k-1}$ and $\chi_{k+1}$, denoted by $\chi_{k-1} \parallel \chi_{k+1}$. C{\scriptsize HI}C{\scriptsize HI} serves as the core primitive in C{\scriptsize HI}L{\scriptsize OW}, which constitutes a family of tweakable block ciphers along with an associated pseudorandom function (PRF) specifically designed for embedded code encryption.

We conduct a comprehensive comparative analysis of $\chi_{n,m}$, $\chi_n$, $\cchi_n$ regarding cryptographic security and implementation costs for small parameter values. We prove Conjecture~1 in \cite{belkheyar2025chi} which determines the algebraic degree of $\cchi_n^{-1}$ as a by-product. {It shows that  $\chi_{n,m}$ has a higher algebraic degree than $\chi_n$, $\cchi_n$, and the inverse of $\chi_{n,m}$ also has a higher algebraic degree than $\chi_n$, $\cchi_n$}. Nevertheless, implementation cost analysis reveals comparable hardware requirements among 
$\chi_{n}$, $\cchi_n$, and a modified {$\chi_{n,3}$-variant}. Tables~\ref{tb_bs}-\ref{tb_ws} present detailed comparisons of critical security metrics including differential uniformity, nonlinearity, boomerang uniformity, and differential-linear uniformity of $\chi_{n,3}, \chi_n$, and $\cchi_n$ for $n=5, 6, 8$. The data demonstrates a nuanced tradeoff: while our proposed functions exhibit superior performance in certain security indicators (e.g., boomerang/differential-linear resistance), they show slight disadvantages in others (notably differential uniformity and nonlinearity). These findings confirm our construction as a viable alternative to $\chi_n$ and $\cchi_n$ for odd- and even-dimensional vector spaces, respectively. We anticipate that this expanded design space will enhance flexibility in cryptographic algorithm development, particularly for resource-constrained applications.

The remainder of this paper is organized as follows. Section~\ref{sec:prelim} introduces the necessary preliminaries and notations. Section~\ref{sec3} investigates the properties of the mappings $\theta_{m,k}$ and constructs an {abelian group $G_{n,m}$} comprising invertible mappings from $\F_2^n$ to itself. Section~\ref{sec4} applies these theoretical insights to analyze specific mappings within $G_{n,m}$, including $\chi_{n,m}$.  Section~\ref{sec:com} provides a database detailing cryptographic properties of  the compositional iterates of $\chi_{n,m}$ and conducts a comparative security and implementation cost analysis for small values of $n$ and $m$.
Section~\ref{sec:conclusion} concludes this paper.

\section{Notations and preliminaries}\label{sec:prelim}

Let $[n]$ denote the integer set $\{0, 1,\ldots,n-1\}$.
Let $\mathcal{M}_n$ denote the set of all mappings from $\F_2^n$ to itself, and let $S \in \mathcal{M}_n$ be the cyclic left shift operator over $\F_2^n$, defined as $S(x_0,x_1,\ldots,x_{n-1})=(x_1,\ldots,x_{n-1},x_0)$. It is straightforward to verify that {$S$ is an invertible linear mapping}. For $i \ge 0$, let $S^i = \underbrace{S \circ S \circ \cdots \circ S}_{i}$ denote the $i$-th iterate of $S$, with $S^0$ representing the identity mapping on $\F_2^n$. In order to generate a nonlinear mapping
over $\F_2^n$ from the shift operator $S$, recall the Hadamard product $\odot$ over $\F_2^n$ as follows: $x \odot y\colon=(x_0y_0,x_1y_1,\ldots,x_{n-1}y_{n-1})$, where $x=(x_0, x_1,\ldots, x_{n-1}), y=(y_0, y_1, \cdots, y_{n-1})\in \F_2^n$. It is easy to see that the Hadamard product $\odot$ is commutative and distributive with respect to addition, i.e., $x \odot y= y \odot x$ and  $x \odot (y+z) = x \odot y + x \odot z$ for all $x,y,z\in \F_2^n$. In particular,  $x \odot x= x$ and $x\odot ({\bf 1} +x) ={\bf 0}$, where ${\bf 1} =(1, 1, \cdots, 1)$ and ${\bf 0}= (0, 0, \cdots, 0)$
are the all-one and all-zero vectors in $\F_2^n$, respectively.

Furthermore, we formally extend the Hadamard product to the set $\mathcal{M}_n$ as follows:
\begin{align*}
	&\mathcal{M}_n \times \mathcal{M}_n \to \mathcal{M}_n \\
    &(f, g) \mapsto f \odot g, \,\, {\rm where } \,\,  {(f\odot g)(x) = f(x) \odot g(x)}
\end{align*}
 for $x \in \F_2^n$.
More precisely, let the $i$-th coordinate functions of $f$ and $g$ be $f_i$ and $g_i$ respectively, then the $i$-th coordinate functions of $f \odot g$ is represented as
\begin{equation*}
	{(f \odot g)_i(x) = f_i(x) \cdot g_i(x)}
\end{equation*}
for $i \in [n]$.
The set $\mathcal{M}_n$ is equipped with three fundamental operations: addition ($+$), composition 
($\circ$), and Hadamard product ($\odot$) of mappings over $\F_2^n$. For any $f, g, h\in \mathcal{M}_n $, it is straightforward to verify that $f\odot (f +\1)={\bf 0}$. Moreover, these operations satisfy the following algebraic properties:

The composition $\circ$ is right distributive with respect to $+$ and $\odot$, i.e.,
\[ (f + g)\circ h = f \circ h +g\circ h, \,\,\, (f\odot g) \circ h=(f \circ h)\odot(g \circ h) .\]

The Hadamard product $\odot$ is commutative, associative and distributive with respect to addition, i.e.,
\begin{equation*} \label{eq-com&dis}
\begin{aligned}
&f \odot g = g \odot f,\,\,\,  (f \odot g) \odot h = f \odot (g \odot h),\,\,\, f \odot (g+h) = f \odot g + f \odot h.
\end{aligned}
\end{equation*}

{Note that} the left shift operator $S$ is distributive with respect to $+$ and $\odot$, i.e.,
  \begin{equation*}
\begin{aligned}
(f\odot g)\circ S&=(f\circ S)\odot (g \circ S), \,\,\, S\circ (f \odot g)=(S\circ f) \odot (S \circ g), \\
(f+ g)\circ S&=(f \circ S) + (g \circ S), \,\,\, S\circ (f + g)=(S\circ f) + (S \circ g).
\end{aligned}
  \end{equation*}

To analyze the mapping $\chi_{n,m}$ introduced in (\ref{eq-chinm}), we define a key mapping $ \theta_{m,k}\colon \F_2^n \rightarrow \F_2^n$  as follows:
\begin{equation}\label{eq-gamma}
	\theta_{m,k}(x_0,x_1,\ldots,x_{n-1})=(y_0,y_1,\ldots,y_{n-1}),
\end{equation}
where $m$ and $k$ are nonnegative integers with $m\nmid n$, and the $i$-th coordinate of the output is given by
\begin{equation*}
	y_i=x_{i+mk} \prod_{\substack{j=1,\,\, m \nmid j}}^{mk-1}  \left(x_{i+j}+1\right), \,\,{\rm for }\,\, i\in [n].
\end{equation*}
This mapping can be compactly expressed using the shift operator $S$ and the Hadamard product as
\begin{equation*}
	\theta_{m,k}= S^{mk} \bigodot_{\substack{j=1,\,\, m \nmid j}}^{mk-1} (S^{j}+\mathbf{1} ),
\end{equation*}
where $\mathbf{1}=(1,1,\ldots,1) \in \F_2^n$. 
This formulation illustrates the interplay between the shift operator and the Hadamard product in constructing nonlinear mappings over $\F_2^n$.

\section{Algebraic structure of the set associated with $\theta_{m,k}$}\label{sec3}

In this section, we investigate the vector space generated by $\theta_{m,k}$ over $\F_2$ for $k=1, 2, \ldots, \left \lfloor \frac{n}{m} \right \rfloor$, as well as the associated abelian group, where $m$ is a fixed positive integer satisfying 
$m\nmid n$. To this end, we first establish some fundamental properties of $\theta_{m,k}$.

\begin{lemma}\label{lem:0}
If $mk>n$, then $\theta_{m,k}=\mathbf{0}$,  i.e., $\theta_{m,k}(x)=\mathbf{0}$ for any $x \in \F_2^n$.
\end{lemma}
\begin{proof}
Assume $mk>n$. By the periodicity of the shift operator $S$, we have $S^{mk}=S^{mk-n}$. Since $m \nmid n$, it follows that $m \nmid (mk-n)$. Consequently, the expression for $\theta_{m,k}$ includes a term of the form $(S^{mk-n}+{\bf 1})$. Utilizing the properties of the shift operator $S$ and the Hadamard product, we derive that
\begin{align*}
\theta_{m,k} = S^{mk} \bigodot_{\substack{i=1,\,\, m \nmid i}}^{mk-1} (S^{i}+\mathbf{1}) 
= S^{mk-n}\odot(S^{mk-n}+\1) \bigodot_{\substack{i=1\\
m \nmid i \text{ and } i\neq mk-n}}^{mk-1} (S^{i}+\mathbf{1})
 = {\bf 0}.
\end{align*}
The last equality follows from the fact that $S^{mk-n}\odot(S^{mk-n}+\1)={\bf 0}$.
\end{proof}

Lemma~\ref{lem:0}  immediately leads to the following corollary.

\begin{corollary}
Let $\ell=\left \lfloor \frac{n}{m} \right \rfloor$. Then $k=\ell+1$ is the smallest integer for which $\theta_{m,k}={\bf 0}$.
\end{corollary}

It is important to note that for $1\leq i\leq \ell$, the mappings $\theta_{m,i}$ are constructed as the Hadamard product of certain linear mappings, resulting in nonlinear mappings over $\F_2^n$. Recall that the algebraic degree of a mapping from 
$\F^n_2$ to itself is defined as the maximum multivariate degree of its component functions. This property will be explored in the subsequent analysis.

\begin{lemma}\label{lem:thetaindepent}
Let $\ell=\left \lfloor \frac{n}{m} \right \rfloor$. Then $\deg(\theta_{m,k}(x))=(m-1)k+1$ for $1\leq k\leq \ell$.
In particular, the mappings $\theta_{m,0}, \theta_{m,1}, \cdots, \theta_{m,\ell}$ are linearly independent over $\F_2$, and thus
$\Theta_{n,m}\colon=\left\{\sum\limits_{k=0}^{\ell} a_k\theta_{m,k}\colon a_k \in \F_2 \right\}$ forms a vector space over $\F_2$ of dimension $\ell+1$.
\end{lemma}

\begin{proof}
For $k \ge 1$, the entries of $\theta_{m,k}(x)$ are defined by
\begin{equation*}
		y_i=x_{i+mk} \prod_{\substack{j=1,\,
m \nmid j}}^{mk-1}  \left(x_{i+j}+1\right).
	\end{equation*}
Since all variables $x_i$ appearing in the product are distinct, it follows that $\deg(\theta_{m,k}(x))=(m-1)k+1$. Consequently,  the mappings $\theta_{m,0}, \theta_{m,1}, \cdots, \theta_{m,\ell}$ are linearly independent over $\F_2$.
\end{proof}

Let $\theta_0=\theta_{m,0}$ denote the identity mapping from $\F_2^n$ to itself. From Lemma~\ref{lem:thetaindepent} we know that $\theta_0, \theta_{m,1}, \cdots, \theta_{m,\ell}$ are linearly independent over $\F_2$. To construct a group comprising invertible mappings on $\F_2^n$, we define a set
\begin{equation}\label{eq:group}
{G_{n,m}}=\left\{\theta_0 +\sum\limits_{k=1}^{\ell} a_k\theta_{m,k} \colon a_k \in \F_2 \right\} \subseteq \Theta_{n,m},
\end{equation}
where $m$ is a fixed number with $m\nmid n$ and $\ell= \lfloor \frac{n}{m}\rfloor$.
Next we will show that the set ${G_{n,m}}$ forms a monoid under composition. Recall that a \emph{monoid} $(\mathcal{G},*)$ is a set $\mathcal{G}$ equipped with an associative binary operation $*$ and a neutral element with respect to $*$.
A monoid in which every element has an inverse element is called a \emph{group}. To establish this result, we require the following key lemma.

\begin{lemma} \label{le-st-m}
Let $m$ and $t$ be integers with $m \mid t$, $f=\sum\limits_{i=0}^{k} a_i\theta_{m,i}$, $g=\theta_0 +\sum\limits_{j=1}^{s} b_j\theta_{m,j}$, where $a_i,b_j \in \F_2$. Then,
\begin{equation} \label{eq-stm}
		(S^t \circ f) \bigodot_{v \in \{1,2,\ldots,m-1\}}  (S^{t-v}\circ g+\1) =S^{t-m}\circ \left(\sum\limits_{i=0}^{k} a_i\theta_{m,(i+1)}\right).
\end{equation}
\end{lemma}

\begin{proof}
Let $g_0=\sum_{j=1}^{s} b_j\theta_{m,j}$. We begin by expanding the left-hand side of (\ref{eq-stm}):
\begin{equation}\label{eq:fg}
\begin{aligned}
 &(S^t \circ f)\bigodot\limits_{v \in \{1,2,\ldots,m-1\}}  \left( S^{t-v}\circ g+\1 \right) \\ \\
=&(S^t \circ f)\bigodot_{v \in \{1,2,\ldots,m-1\}}  \left( S^{t-v}\circ \theta_0+\1+S^{t-v}\circ g_0 \right) \\ \\
=&(S^t \circ f)\Bigg(\bigodot_{v \in \{1,2,\ldots,m-1\}} ( S^{t-v}\circ \theta_0+\1)
+ \sum_{\emptyset\neq V\subseteq \{1,2,\ldots,m-1\}} \, \bigodot_{u \in V} (S^{t-u} \circ g_0) \\
&\quad \bigodot_{w \in \{1,2,\ldots,m-1\}\setminus V} (S^{t-w} \circ \theta_0+\1)\Bigg)\\ \\
=&(S^t \circ f)\bigodot_{v \in \{1,2,\ldots,m-1\}}  \big(S^{t-v}+\1 \big)\\
&\quad+ \sum_{\emptyset\neq V\subseteq \{1,2,\ldots,m-1\}} \, \bigodot_{u \in V}  (S^t \circ f) \odot(S^{t-u}\circ g_0) \bigodot_{w \in \{1,2,\ldots,m-1\} \setminus V} \left(S^{t-w}\circ \theta_0+\1\right).
\end{aligned}
\end{equation}
Next, we will show that for any $v \in \{1,2,\ldots,m-1\}$,
$$ (S^t \circ f)\odot (S^{t-v}\circ  g_0)= \sum_{i=0}^k a_i (S^t \circ \theta_{m,i})\odot \sum_{j=1}^{s} b_j (S^{t-v} \circ \theta_{m,j})= {\bf 0}.$$ 
To see this, observe that:
\[\begin{array}{c}
		S^t \circ \theta_{m,i} = S^t \circ \left( S^{mi} \bigodot\limits_{\substack{ u=1\\
m \nmid u}}^{mi-1} ( S^{u}+{\bf 1})\right) = S^{mi+t} \bigodot\limits_{\substack{u=1\\
m \nmid u}}^{mi-1} (S^{t+u}+{\bf 1}), \notag \\  \\
		S^{t-v} \circ \theta_{m,j} =  S^{mj+t-v} \bigodot\limits_{\substack{u=1\\
m \nmid u}}^{mj-1} \left( S^{t-v+u}+{\bf 1}\right). \notag
	\end{array}\]
Since $v \in \{1,2,\ldots,m-1\}$, we have $mj+t-v \ge t+1$ and $mi+t \ge t-v+1$. If $mj-v < mi$, then $t+1 \le mj+t-v < mi+t$. Since $mj-v \not\equiv 0 \pmod m$, the term $(S^{mj+t-v}+{\bf 1})$ must appear in the expression of $S^t \circ \theta_{m,i}$. Since $S^{mj+t-v}\odot(S^{mj+t-v}+{\bf 1})= {\bf 0}$, we conclude that $(S^t \circ \theta_{m,i}) \odot (S^{t-v}\circ \theta_{m,j})= {\bf 0}.$ Similarly, if $mj-v > mi$, the same result holds. Thus, $ (S^t \circ f)\odot (S^{t-v}\circ  g_0)= {\bf 0}$ for all $v \in \{1,2,\ldots,m-1\}$.  Consequently, {by the properties of the operations introduced} in Section~\ref{sec:prelim}, (\ref{eq:fg}) can be reduced to
$$\begin{aligned}
&(S^t \circ f)\bigodot_{v \in \{1,2,\ldots,m-1\}}\left(S^{t-v}\circ g+\1\right)=(S^t \circ f) \bigodot_{v \in \{1,2,\ldots,m-1\}}\left(S^{t-v}+\1 \right)\\
=&\sum_{i=0}^k a_i \left( S^t \circ \theta_{m,i} \right) \bigodot_{v \in \{1,2,\ldots,m-1\}} \left(S^{t-v}+\1 \right)\\
=&\sum_{i=0}^k a_i S^t \circ \left( S^{mi} \bigodot_{\substack{j=1\\
m \nmid j}}^{mi-1} (S^{j}+\1) \right) \bigodot_{v \in \{1,2,\ldots,m-1\}} \left( S^{t-m}\circ\left(S^{m-v}+\1 \right)\right)\\
=&\sum_{i=0}^k a_i S^{t-m} \circ \left( S^{mi+m} \bigodot_{\substack{j=1\\
		m \nmid j}}^{mi-1} (S^{m+j}+\1) \right) \bigodot_{v \in \{1,2,\ldots,m-1\}} \left( S^{t-m}\circ\left(S^{m-v}+\1 \right)\right)\\
=&\sum_{i=0}^k a_i S^{t-m} \circ\left( S^{mi+m}  \bigodot_{\substack{j=m+1\\
m \nmid j}}^{m(i+1)-1} (S^{j}+\1) \bigodot_{v \in \{1,2,\ldots,m-1\}}  \left(S^{v}+\1 \right)\right)\\
=&\sum_{i=0}^k a_i S^{t-m} \circ\left( S^{mi+m} \bigodot_{\substack{j=1\\
m \nmid j}}^{m(i+1)-1} (S^{j}+\1)\right) \\
=&\sum_{i=0}^k a_i S^{t-m}\circ \theta_{m,(i+1)}.
   \end{aligned}$$
\end{proof}
Based on this lemma, we can show that the set ${G_{n,m}}$ introduced in (\ref{eq:group}) forms a monoid under the composition of mappings from $\F_2^n$ to itself.

\begin{theorem} \label{th-close}
Let $f,g \in {G_{n,m}}$, then $f \circ g \in {G_{n,m}}$. In particular, $(G_{n,m}, \circ)$ is a monoid.
\end{theorem}

\begin{proof}
Let $f=\theta_0 +\sum\limits_{i=1}^{\ell} a_i\theta_{m,i}$ and $g=\theta_0 +\sum\limits_{j=1}^{\ell} b_j\theta_{m,j}$ and $\ell=\lfloor \frac{n}{m} \rfloor$. Then we have
\begin{equation}\label{align-close}
\begin{aligned}
f \circ g =& \left(\theta_0 +\sum\limits_{i=1}^{\ell} a_i\theta_{m,i}\right) \circ g  = g + \left(\sum\limits_{i=1}^{\ell} a_i\theta_{m,i}\right) \circ g  \\
=& \theta_0 +\sum\limits_{j=1}^{\ell} b_j\theta_{m,j} +\sum\limits_{i=1}^{\ell} a_i\left(\theta_{m,i} \circ  g\right).
\end{aligned}
\end{equation}
Let $g_k\colon= \sum\limits_{j=0}^{\ell} b_j\theta_{m,(j+k)}$ with $b_0=1$. We have 

\begin{equation}\label{eq:theta}
\begin{aligned}
&(S^{m(i-k)}\circ g_k) \bigodot_{\substack{v=1\\ m \nmid v}}^{m(i-k)-1}(S^v \circ g+\1) \\
&= (S^{m(i-k)}\circ g_k) \bigodot_{\substack{v_1=1\\ m \nmid v_1}}^{m-1}(S^{m(i-k)-v_1}\circ g +\1) \bigodot_{\substack{v_2=1\\ m \nmid v}}^{m(i-k-1)-1}(S^{v_2}\circ g +\1)\\
&= \left(S^{m(i-k-1)} \circ \sum^{\ell}_{j=0} b_j \theta_{m,(j+k+1)}\right) \bigodot_{\substack{v_2=1\\ m \nmid v}}^{m(i-k-1)-1}(S^{v_2}\circ g +\1)\\
&= \left(S^{m(i-k-1)} \circ g_{k+1}\right) \bigodot_{\substack{v=1\\ m \nmid v}}^{m(i-k-1)-1}(S^{v}\circ g +\1),
\end{aligned}
\end{equation}
{where the second equal sign is obtained by Lemma \ref{le-st-m}. By applying \eqref{eq:theta} repeatedly, we have

\begin{equation}\label{eq:theta2}
\begin{aligned}
  \theta_{m,i}\circ g &= \left( S^{mi}\bigodot_{\substack{v=1\\
m \nmid v}}^{mi-1}(S^{v}+\1) \right) \circ g = \left( S^{mi}\circ g \right)\bigodot_{\substack{v=1\\
m \nmid v}}^{mi-1} \left( S^{v}\circ g+\1\right) \\
&= \left(S^{m(i-1)}\circ g_{1}\right)  \bigodot_{\substack{v=1\\m \nmid v}}^{m(i-1)-1} (S^{v}\circ g+\1)\\
&= \left(S^{m(i-2)}\circ g_{2}\right)  \bigodot_{\substack{v=1\\m \nmid v}}^{m(i-2)-1} (S^{v}\circ g+\1)\\
&=\cdots=g_i=\sum\limits_{j=0}^{\ell} b_j\theta_{m,(i+j)}.
\end{aligned}\end{equation}}

Then  (\ref{align-close}) becomes
\begin{align*}
f \circ g =\theta_0 +\sum\limits_{j=1}^{\ell} b_j\theta_{m,j} +\sum\limits_{i=1}^{\ell} a_i\left(\theta_{m,i} \circ  g\right) 
=\theta_0 +\sum\limits_{j=1}^{\ell} b_j\theta_{m,j} +\sum\limits_{i=1}^{\ell} a_i\sum\limits_{j=0}^{\ell} b_j\theta_{m,(i+j)}.
\end{align*}

This shows that $f \circ g\in {G_{n,m}}$. It is easy to see that the neutral element of ${G_{n,m}}$ under the composition is $\theta_0$. Therefore, $({G_{n,m}}, \circ)$ is a monoid.
\end{proof}

In the following, we will prove that $({G_{n,m}}, \circ)$ is an abelian group by demonstrating that $({G_{n,m}}, \circ)$ is isomorphic to an abelian group. Consequently, ${G_{n,m}}$ defines a broad class of permutations of $\F_2^n$.

Let $\ell=\left \lfloor \frac{n}{m} \right \rfloor$. Consider the factor ring $R_\ell=\F_2[z]/(z^{\ell+1})$ and its unit group $R_\ell^*=(\F_2[z]/(z^{\ell+1}))^*$. Define the coset of $f$ as $[f] =f+(z^{\ell+1})$, where $(z^{\ell+1})$ is the ideal generated by $z^{\ell+1}$ in $\F_2[z]$. {Since $\F_2[z]$ is an Euclidean domain, we have $[f] \in R_{\ell}^*$} if and only if $\gcd(f(z),z^{\ell+1})=1$, which is equivalent to the constant term of $f$ being $1$. Moreover, polynomials in $R_\ell$ with a constant term of $1$ constitute exactly half of $R_{\ell}$. Thus, the following lemma is immediate.

\begin{lemma}\label{uniq}
Let $f(z) \in \F_2[z]$, then $[f]$ is a unit of $R_{\ell}$ if and only if the constant term of $f$ is 1. Furthermore, the unit group $R_\ell^*$ has $2^\ell$ elements.
\end{lemma}

To show that ${G_{n,m}}$ is an {abelian group}, we define a mapping $\varphi\colon \Theta_{n,m} \to R_\ell$ as follows:
\begin{equation*}
	\varphi \left(\sum_{i=0}^{\ell}a_i\theta_{m,i}\right) = \sum_{i=0}^{\ell}a_iz^i.
\end{equation*}
It is a straightforward verification that the mapping $\varphi$ is well-defined. The following lemma demonstrates that the restriction of $\varphi$ to $G_{n,m}$ is a monoid homomorphism.

\begin{lemma} \label{le-homo}
Let $f,g \in {G_{n,m}}$. Then $\varphi(f \circ g)=\varphi(f) \cdot \varphi(g)$. Consequently, the restriction of $\varphi$ to $G_{n,m}$ is a monoid homomorphism.
\end{lemma}

\begin{proof}
Let $f=\sum\limits_{i=0}^{\ell }a_i\theta_{m,i}$ and $g=\sum\limits_{i=0}^{\ell}b_i\theta_{m,i}$ be elements of ${G_{n,m}}$, where $a_0=b_0=1$. By (\ref{eq:theta2}) we have
\begin{align*}
\varphi(\theta_{m,k} \circ g) =\varphi \left(\sum_{i=0}^{\ell} {b_i\theta_{m,{(i+k)}}}\right) =\sum_{i=0}^{\ell}b_iz^{i+k} 
= z^k \cdot \sum_{i=0}^{\ell} b_i z^i =\varphi(\theta_{m,k}) \cdot \varphi(g).
\end{align*}
{ It is straightforward to show that $\varphi$ is an additive homomorphism, so we have}
\begin{align*}
\varphi(f \circ g) &=\varphi\left(\sum\limits_{i=0}^{\ell}a_i\theta_{m,i} \circ g\right) =\sum\limits_{i=0}^{\ell}a_i \varphi(\theta_{m,i} \circ g) \\
&=\sum\limits_{i=0}^{\ell}a_i \varphi(\theta_{m,i}) \cdot \varphi(g) =\varphi(f) \cdot \varphi(g).
\end{align*}
It is clear that $\theta_0$ is the neutral element of ${G_{n,m}}$ and $\varphi(\theta_0)=1$ is the neutral element of $R_\ell^*$.
\end{proof}

\begin{theorem} \label{th-iso}
The set ${G_{n,m}}$ is an {abelian group} under the mapping composition, and it is isomorphic to $R_\ell^*$.
\end{theorem}

\begin{proof}
By Lemma~\ref{le-homo}, {the restriction of $\varphi$ to $G_{n,m}$} is a monoid homomorphism. For any $f=[1+\sum\limits_{i=1}^{\ell}a_i z^i] \in R_{\ell}^*$, we have $f =\varphi(\theta_0+\sum\limits_{i=1}^{\ell}a_i\theta_{m,i})$, where $\theta_0+\sum\limits_{i=1}^{\ell}a_i\theta_{m,i} \in {G_{n,m}}$. Thus, $\varphi$ is surjective from $G_{n,m}$ to $R_{\ell}^*$. Moreover, since $|{G_{n,m}}|=|R_\ell^*|=2^\ell$, $\varphi $ is also bijective from $G_{n,m}$ to $R_{\ell}^*$. Therefore, ${G_{n,m}}$ is isomorphic to $R_\ell^*$, and since
$R_\ell^*$ is an {abelian group}, ${G_{n,m}}$ is also an {abelian group}. 
\end{proof}

\begin{remark}
Every mapping $f \in {G_{n,m}}$ is a permutation of $\F_2^n$. To determine the compositional inverse of $f$, it suffices to compute the inverse of $\varphi(f)$ in $R_\ell$, which is relatively straightforward.
\end{remark}

\begin{proposition} \label{prop:inverse}
Let $f=\theta_0+\sum\limits_{i=1}^{\ell}a_i\theta_{m,i} \in {G_{n,m}}$. Then its compositional inverse is
 $$ f^{-1}=\theta_0+\sum\limits_{j=1}^{\ell}b_j\theta_{m,j},$$
where $b_1=a_1$, $b_j=a_j+\sum \limits_{\begin{subarray}{c}
				u+v=j \\
				u,v \in \{1,2,\ldots, j-1\}
\end{subarray}}a_ub_v$ for $j \in \{2,3,\ldots,\ell\}$. Furthermore, $f$ is an involution if and only if $a_i=0$ for any $i \in \{1,2,\ldots,\left \lfloor \ell/2 \right \rfloor\}$, i.e., the lowest degree of the non-constant terms of $f$ is at least $\lfloor \ell/2 \rfloor+1$.

\end{proposition}

\begin{proof}
By Theorem \ref{th-iso}, it suffices to find the multiplicative inverse of $\varphi(f)=\left[1+\sum\limits_{i=1}^{\ell}a_iz^{i}\right]$. Assume $\varphi(f^{-1})=\left[1+\sum\limits_{j=1}^{\ell}b_jz^{j}\right]$. Then, 
\begin{equation}
	\begin{aligned}
		1 &=\varphi(f) \cdot \varphi(f^{-1}) =\left[ (1+\sum\limits_{i=1}^{\ell}a_iz^{i})(1+\sum\limits_{j=1}^{\ell}b_jz^{j}) \right] \\
		&\equiv \left[1+\sum\limits_{i=1}^{\ell}a_iz^{i}+\sum\limits_{j=1}^{\ell}b_jz^{j}+\sum_{t=2}^{\ell}c_tz^{t}\right],     \label{align-inverse}
	\end{aligned}
\end{equation}

where
\begin{equation*}
		c_t=
			\sum_{\substack{
					u+v=t \\
					u,v \in \{1,2,\ldots,\ell-1\}}
			}
			a_ub_v.
\end{equation*}
Let $c_1=0$. Then the above congruence holds if and only if $a_t+b_t+c_t=0$ for all $t \in \{1,2,\ldots,\ell\}$. Consequently,  we have	
\begin{equation*}
\left\{\begin{array}{lcl}
b_1 =a_1, \\	
 b_t=a_t+\sum\limits_{\substack{
				u+v=t \\
				u,v \in \{1,2,\ldots,t-1\}}}
	a_ub_v, & {\rm if} \,\,\, 2 \le t \le \ell.
\end{array}\right.
\end{equation*}
When $\ell=1$, since $a_1=b_1$, it follows that $f$ is an involution. For $\ell \ge 2$, $f$ is an involution if and only if
\begin{equation*} \label{eq-ch}
    	0=\begin{array}{ll}
    		\sum\limits_{\begin{subarray}{c}
    				u+v=h \\
    				u,v \in \{1,2,\ldots,h-1\}
    		\end{subarray}}
    		a_u a_v
    	\end{array}
            =\left\{\begin{array}{rl}
            		a_{h/2}, & {\rm if} \,\,\,  2 \mid h \\
            		0, & {\rm if} \,\,\, 2 \nmid h,
            	\end{array}\right.
\end{equation*}
for all $h \in \{1,2,\ldots, \ell\}$. This condition holds if and only if $a_i=0$ for all $i \in \{1,2,\ldots,\left \lfloor \ell/2 \right \rfloor\}$.
\end{proof}

\section{The properties of the mappings $\chi_{n,m}$ and its iterates} \label{sec4}
In this section, we delve into the mapping $\chi_{n,m}$ defined by~\eqref{eq-chinm} and its iterates. Specifically, we give an explicit expression for the inverse of $\chi_{n,m}$ and characterize the algebraic structure of its iterates, 
This characterization offers deeper insights into the fixed points and cycle structure of the mapping $\chi_{n,m}$.

\begin{lemma}\label{lem:order}
Let $\ell =\lfloor\frac{n}{m}\rfloor$ and $f=\left[1+z^j+\sum\limits_{i=j+1}^{\ell}a_iz^i\right] \in {R_\ell}^*$, where $z^j$ is the non-constant term with the lowest degree in $f$. Then, the order of $f$ in $R_\ell^*$ is $2^{\lceil \log_2\frac{\ell+1}{j}\rceil}$, i.e., $\ord(f) =2^{\lceil \log_2\frac{\ell+1}{j}\rceil}$.
\end{lemma}

\begin{proof}
Since $|R_{\ell}^*|=2^\ell$, the order of $f$ is a power of 2, denoted by $2^r$, by Lagrange's Theorem. We then have
\begin{equation}
\left\{\begin{aligned}\label{eq:orderf}
 1 &=f^{2^r}= \left[ 1+z^j+\sum\limits_{i=j+1}^{\ell}a_iz^i\right]^{2^r}=\left[ 1+z^{2^rj}+\sum\limits_{i=j+1}^{\ell}a_iz^{2^ri}\right],\\
 1 &\neq f^{2^{r-1}}=\left[ 1+z^{2^{r-1} j}+\sum\limits_{i=j+1}^{\ell}a_iz^{2^{r-1}i}\right].
\end{aligned}\right.
\end{equation}
Note that (\ref{eq:orderf}) holds if and only if $2^rj \ge \ell+1 >2^{r-1}j$. Therefore, we have $2^{r-1} <\frac{\ell+1}{j} \le 2^r$, which imples that $\ord(f) =2^{\lceil \log_2\frac{\ell+1}{j}\rceil}$.
\end{proof}

Due to {their simple expression}, the binomial cases are of particular interest. Using Proposition~\ref{prop:inverse}, we derive the explicit expression for the inverse of the binomial $\theta_0+\theta_{m,k} \in {G_{n,m}}$.

\begin{theorem}\label{thm:bin}
Let $k \ge1$, $\ell=\left \lfloor \frac{n}{m} \right \rfloor $ and $s=\left \lfloor \frac{\ell}{k} \right \rfloor$. Let $f=\theta_0 +\theta_{m,k} \in {G_{n,m}}$. Then $\ord(f) =2^{\lceil \log_2(\frac{\ell+1}{k})\rceil}$, and the compositional inverse of $f$ is
	\begin{equation*}
		f^{-1} =\theta_0 +\theta_{m,k} +\theta_{m,2k} + \cdots +\theta_{ m, sk}.
	\end{equation*}
Moreover, $\deg(f)=(m-1)k+1$ and $\deg(f^{-1})=(m-1)sk+1$.
\end{theorem}

\begin{proof}
Let $a_k=1$ and $a_i =0$ for $i \neq k$. By Proposition~\ref{prop:inverse}, we have $b_t=0$ for $t <k$ and $b_k=a_k=1$. In the following, we assume $\ell \geq t >k$. Let $t'<k$ with $t\equiv t' \pmod{k}$. If $k \nmid t$, then
$b_t=a_t+a_kb_{t-k}=b_{t-k}=b_{t'}=0$.  If $k \mid t$, then $b_t=a_t+a_kb_{t-k}=b_{t-k}=b_k=1$. Therefore, $b_t=1$ if and only if $t$ is a multiple of $k$ and does not exceed $\ell$. The $\deg(f)$ and $\deg(f^{-1})$ follow from Lemma~\ref{lem:thetaindepent}, the order of $f$ is obtained by Lemma~\ref{lem:order}.
\end{proof}

\begin{remark}
By examining the state diagram of $\chi_n$, Schoone and Daemen in~\cite{schoone2024state} proved that $\ord(\chi_n)=2^{\lceil {\rm log}_2 \frac{n+1}{2} \rceil}$ when $n$ is odd. This is a special case of Theorem~\ref{thm:bin} for 
$m=2$, $k=1$ and $n$ being odd. So,  Theorem~\ref{thm:bin} generalizes Corollary~9 of \cite{schoone2024state}.
\end{remark}

Let $k=1$ in Theorem~\ref{thm:bin}, and we have $\chi_{n,m}=\theta_0+\theta_{m,1}$. As previously stated, $\chi_{n,m}$ permutes $\F_2^{n}$ if $m \nmid n$. However, if $m \mid n$, then we have
\begin{equation*}
\chi_{n,m}(1,\mathbf{0}_{m-1}, 1,\mathbf{0}_{m-1}, \cdots, 1,\mathbf{0}_{m-1}) ={\bf 0} =\chi_{n,m}({\bf 0}),
\end{equation*}
where $\mathbf{0}_i$ denotes the zero vector with $i$ bits. Thus, $\chi_{n,m}$ is not a permutation.
This leads us to the necessary and sufficient condition for $\chi_{n,m}$ being a permutation, along with its compositional inverse.

\begin{corollary} \label{cor-chinm}
Follow the notation defined above. The mapping $\chi_{n,m}$ in (\ref{eq-chinm}) is a permutation of $\F_2^n$ if and only if $m \nmid n$. In particular, if $m \nmid n$, then the inverse of $\chi_{n,m}$ is
\begin{equation*}
		\chi_{n,m}^{-1} =\theta_0 +\theta_{m,1} +\theta_{m,2} + \cdots +\theta_{ m,\ell},
\end{equation*}
or equivalently,
\begin{equation} \label{eq-inverse}
		y_i=x_i +\sum_{k=1}^{\ell}x_{i+mk} \prod_{\substack{j=1\\ m \nmid j}}^{mk-1} \left(x_{i+j}+1\right)
\end{equation}
for the $i$-th coordinate of the output, where $\ell=\left \lfloor \frac{n}{m} \right \rfloor$. Moreover, $\deg(\chi_{n,m}^{-1})=(m-1)\ell+1$.
\end{corollary}

\begin{remark}
By considering an affine variety associated with $\chi_n$, Liu et al. in \cite{liu2022inverse} provided the specific expression for $\chi_n^{-1}$ as
\begin{equation} \label{eq-liu}
		y_i=x_i +\sum_{k=1}^{\ell}x_{i-2k+1} \prod_{j=k}^{\ell}  \left(x_{i-2j}+1\right)
\end{equation}
for the $i$-th coordinate of the output, where $n$ is odd and $\ell=(n-1)/2$. In fact, the equality {\rm (\ref{eq-inverse})} is consistent with {\rm (\ref{eq-liu})} for $m=2$.
When $m=2$, {\rm (\ref{eq-inverse})} becomes
\begin{align*}
		y_i &= x_i +\sum_{k=1}^{\ell}x_{i+2k} \prod_{j=0}^{k-1}  \left(x_{i+2j+1}+1\right) \\
		&= x_i +\sum_{k'=1}^{\ell}x_{i+2(\ell-k'+1)} \prod_{j'=\ell-k+1}^{\ell}  \left(x_{i+2(\ell-j')+1}+1\right) \\
		&= x_i +\sum_{k'=1}^{\ell}x_{i-2k'+1} \prod_{j'=k'}^{\ell}  \left(x_{i-2j'}+1\right),
\end{align*}
where the second equality is obtained by substitution ${k'=\ell-k+1},j'=\ell-j$, and the third equality holds since $2\ell + 1 = n$. Therefore, Corollary \ref{cor-chinm} generalizes Theorem~1 in~\cite{liu2022inverse}.
\end{remark}

Next, we investigate the fixed points of the iterate $\chi_{n,m}^k$ of $\chi_{n,m}$ for some positive integer $k$, which characterize some cryptographic properties. To this end, we need to determine the explicit expression of the iterate $\chi_{n,m}^k$. Let $j, k \in \mathbb{N}$ be written in base 2 as $j=j_0+2j_1+2^2j_2+ \cdots +2^{s-1}j_{s-1}$, and $k=k_0+2k_1+2^2k_2+ \cdots +2^{s-1}k_{s-1}$, where $j_i,k_i \in \F_2$ for any $i \in [s]$. Denote $j \preceq k$ if $j_i \le k_i$ for any $i \in [s]$.

\begin{theorem} \label{th-iterate}
	Let $k \ge 1$ and $\ell=\left \lfloor \frac{n}{m} \right \rfloor $.  Then
	\begin{equation*}
		\chi_{n,m}^k=\sum_{j=0}^{\min\{k,\ell\}}a_j\theta_{m,j},
	\end{equation*}
where $a_j=1$ if and only if $j \preceq k$.
\end{theorem}

\begin{proof}
It is known that ${G_{n,m}}$ is isomorphic to $R_\ell^*$ by the mapping $\varphi$. We obtain $\chi_{n,m}^k$ by investigating the inverse of $\varphi(\chi_{n,m}^k)$. Since $\varphi(\chi_{n,m}^k)= (1+z)^k\in R_\ell^*$, we need to expand the polynomial $(1+z)^k$ to find its preimage under $\varphi$.

By Lucas's theorem,
\begin{equation*}
		\binom{k}{j} = \binom{k_0+2k_1+2^2k_2+ \cdots +2^{s-1}k_{s-1}}{j_0+2j_1+2^2j_2+ \cdots +2^{s-1}j_{s-1}}
		\equiv \binom{k_0}{j_0} \cdot \binom{k_1}{j_1} \cdot \, \cdots \, \cdot \binom{k_{s-1}}{j_{s-1}} \pmod 2.
\end{equation*}
Thus, $\binom{k}{j}\equiv 1 \,\, {\rm mod}\,\, 2$ if and only if $j \preceq k$. Therefore,
\begin{equation*}
(1+z)^k =\sum_{j=0}^{k}\binom{k}{j}z^j =\sum_{j=0}^{k}a_jz^j \equiv \sum_{j=0}^{\min\{k,\ell\}}a_jz^j \,\, \, {\rm mod} \,\, \, z^{\ell+1},
\end{equation*}
where $a_j=1$ if and only if $j \preceq k$.	 Consequently, the preimage of $(1+z)^k$ under $\varphi$ is $\sum_{j=0}^{\min\{k,\ell\}}a_j\theta_{m,j}$. This completes proof.
\end{proof}

A vector $x \in \F_2^n$  is called a \textit{fixed point} of $\chi_{n,m}^k$ if $\chi_{n,m}^k(x)=x$. 
Note that $\chi_{n,m}^k(x)=x$ for $x \in \{\bf 0,1\}$ and some positive integer $k$. Therefore, we focus on nontrivial fixed points.

\begin{theorem} \label{th-fixed point}
Let $\ell=\lfloor\frac{n}{m}\rfloor$. The mapping $\chi_{n,m}^k$ has a nontrivial fixed point if and only if $k=2^j$ for some $j \in \{1,2,\ldots,r-1\}$, where $r={\lceil \log_2({\ell+1})\rceil}$. Furthermore, $x \in \F_2^n$ is a fixed point of $\chi_{n,m}^{2^j}$ if and only if $x$ does not contain a substring of the form $(\mathbf{0}_{m-1},*,\mathbf{0}_{m-1},*, \cdots ,\mathbf{0}_{m-1},1)$ of length $2^jm$, where $* \in \F_2$ is any bit value and $\mathbf{0}_{m-1}$ denotes $m-1$ consecutive bits of $0$.
	
\end{theorem}

\begin{proof}
By Theorem~\ref{thm:bin}, we know that $\ord(\chi_{n,m})=2^r$, which is the least common multiple of the lengths of all cycles in $\chi_{n,m}$. Thus, the length of any cycle in $\chi_{n,m}$  must be a power of 2. By Theorem \ref{th-iterate} we have $\chi_{n,m}^{2^j}=\theta_0+\theta_{m,2^j}$. Consequently, $x$ is a fixed point of $\chi_{n,m}^{2^j}$  if and only if $\theta_{m,2^j}(x)=\mathbf{0}$.
Note that the $i$-th component of $\theta_{m,2^j}(x)$ is equal to $1$ if and only if $(x_{i+1}, \cdots, x_{i+2^j m})= (\mathbf{0}_{m-1},*,\mathbf{0}_{m-1},*, \cdots ,\mathbf{0}_{m-1},1)$. This completes proof.
\end{proof}

\begin{remark}
The fixed points of the iterates of $\chi_{n,m}$ essentially determine the cycle structure of $\chi_{n,m}$. Specifically, a fixed point $x \in \F_2^n$ of $\chi_{n,m}^j$ lies in a cycle of length dividing $j$ in the cycle decomposition of the mapping $\chi_{n,m}$.  Moreover, if $x$ is not a fixed point of $\chi_{n,m}^{j'}$ for any $j'<j$, then the length of the cycle containing $x$ is exactly $j$.
\end{remark}

Recall that $\chi_{n,m}=\theta_0+\theta_{m,1}$, and by Theorem~\ref{th-iterate},   $\chi_{n,m}^2=\theta_0+\theta_{m,2}$. Thus, $\chi_{n,m}$ is an involution if and only if $\chi_{n,m}^2(x)=x$ for any $x \in \F_2^n$, i.e., $\theta_{m,2}= {\bf 0}$. By {Lemma}~\ref{prop:inverse},  we have the following result.
	
\begin{corollary}\label{cor}
	$\chi_{n,m}$ is an involution if and only if $n \le 2m-1$.
\end{corollary}

\begin{remark}
By Corollary~\ref{cor}, {the mapping $\chi_n$ can be an involution only if the dimension $n \leq 3$.} However, $\chi_{n,m}$ can be involutions for large dimension $n$. 
\end{remark}

Based on the theorems presented in this paper, we provide an example which is consistent with numerical experiments conducted on Magma (see \cite{MR1484478}).

\begin{example}
{Let $\chi_{8,3}\colon\F_2^8 \to \F_2^8$ be the mapping given by $y=\chi_{8,3}(x)$, where the $i$-th coordinate of $y$ is represented as
	\begin{equation*}
		y_i=x_i+x_{i+3}(x_{i+2}+1)(x_{i+1}+1).
	\end{equation*}
By Theorem \ref{thm:bin}, we have $\ord(\chi_{8,3}) =4$. And the compositional inverse of $\chi_{8,3}$ is $\chi_{8,3}^{-1} =\theta_0 +\theta_{3,1} +\theta_{3,2}$ by Corollary \ref{cor-chinm}, where the $i$-th coordinate of the output $y$ is represented as
	\begin{align*}
		y_i= \ &x_i+x_{i+3}(x_{i+2}+1)(x_{i+1}+1)
		+x_{i+6}(x_{i+5}+1)(x_{i+4}+1)(x_{i+2}+1)(x_{i+1}+1)\\
		=\ & x_i+(x_{i+2}+1)(x_{i+1}+1)
		(x_{i+3}+x_{i+6}(x_{i+5}+1)(x_{i+4}+1)).
	\end{align*}
	Moreover, $\deg(\chi_{8,3})=3$, $\deg(\chi_{8,3}^{-1})=5$.
	By Theorem \ref{th-iterate}, $\chi_{8,3}^2=\theta_0 +\theta_{3,2}$, where the $i$-th coordinate of the output $y$ is represented as
	\begin{equation*}
		y_i= x_i+x_{i+6}(x_{i+5}+1)(x_{i+4}+1)(x_{i+2}+1)(x_{i+1}+1),
	\end{equation*}
	and $(1,1,1,1,1,1,1,1) \in \F_2^8$ is a fixed point of $\chi_{8,3}^2$ by Theorem \ref{th-fixed point}. }
\end{example}

\section{A comparative analysis of the $\chi_n$, $\chi_{n,m}$ and $\rcchi_n$ functions on a small scale}\label{sec:com}

This section establishes a comparative evaluation for the $\chi_n$, $\chi_{n,m}$ and $\rcchi_n$ functions, focusing on their cryptographic security properties and implementation costs at reduced parameter scales (specifically for small values of $n$). To assess these characteristics, we first formalize key cryptographic security metrics, {including algebraic degree, differential uniformity, nonlinearity, boomerang uniformity, differential-linear uniformity.}

\subsection{Cryptographic security metrics for vectorial Boolean functions}\label{subsec:security_metrics}

To ensure robustness against modern cryptanalytic techniques, vectorial Boolean functions used in cryptography must satisfy rigorous security criteria. We focus on four fundamental cryptographic properties: differential uniformity, algebraic degree, nonlinearity, boomerang uniformity, and differential-linear uniformity. These metrics collectively characterize a function's resistance to major attack vectors in symmetric cryptography.

Differential cryptanalysis, introduced by Biham and Shamir~\cite{DBLP:journals/joc/BihamS91}, remains a cornerstone in block cipher cryptanalysis. Following Nyberg's formalization~\cite{DBLP:conf/eurocrypt/Nyberg93}, 
let $F(x)$ be a function from $\F_2^n$ to itself. The \emph{differential uniformity} of $F$ is defined as:
\begin{equation*}
	\Delta_F = \max_{\substack{a \in \F_2^n\setminus\{\mathbf{0}\}, \ b \in \F_2^n}} \delta_F(a,b),
\end{equation*}
where $\delta_F(a,b)=|\{x \in \F_2^n\colon F(x+a)+F(x)=b\}|$. The complete differential behavior is captured by the \emph{differential spectrum}, which is a multiset:
\begin{equation*}
	\left\{ \delta_F(a,b) \colon a \in \F_2^n\setminus\{\mathbf{0}\},\ b \in \F_2^n \right\}. 
\end{equation*}

Resistance to linear cryptanalysis~\cite{DBLP:conf/eurocrypt/Matsui93} is quantified through nonlinearity. For a vector Boolean function $F(x)$ from $\F_2^n$ to itself, the \emph{nonlinearity} of $F(x)$ is:
\begin{equation*}
	\text{NL}_F = 2^{n-1}-\frac{1}{2} \max_{\substack{a \in \F_2^n\setminus\{\mathbf{0}\}, \ b \in \F_2^n}} |{W_F(a,b)}|,
\end{equation*}
{where $W_F(a,b)=\sum_{x \in \F_{2}^n}(-1)^{a\cdot F(x)+b\cdot x}$ is the Walsh transform,}
and $\cdot$ denotes the standard inner product. The \emph{Walsh spectrum} comprises all Walsh coefficients as a multiset:
\begin{equation*}
	\left\{ {W_F(a,b)} \,\colon\, a, b \in \F_2^n \right\}.
\end{equation*}

To address advanced differential variants, Cid et al.~\cite{cid2018boomerang} introduced boomerang uniformity. For $a,b\in \F_{2^n}$, let $\beta_F(a,b)$ count solutions to:
\begin{equation*}
	F^{-1}(F(x)+b) + F^{-1}(F(x+a)+b) = a.
\end{equation*}
The \emph{boomerang uniformity} is:
\begin{equation*}
	\mathrm{B}_F = \max_{\substack{a, b\in \F_2^n\setminus\{\mathbf{0}}\} } \beta_F(a,b),
\end{equation*}
with the \emph{boomerang spectrum} given by the multiset:
\begin{equation*}
	\left\{\beta_F(a,b)\, \colon\, a,b \in \F_2^n\setminus\{\mathbf{0}\} \right\}.
\end{equation*}

Recent work by Tang et al.~\cite{tang2021construction} formalized resistance to differential-linear attacks through the Differential-Linear Connectivity Table (DLCT) of a {vectorial Boolean function} $F(x)$ from $\F_2^n$ to itself. This is {a $2^n \times 2^n$ table}, 
whose rows correpond to input differences of $F$ and whose columns correspond to bit masks of outputs of $F$. For $a, b\in \F_2^n$, the value at position $(a,b)$ in the DLCT of $F$ is
\begin{equation*}
	\text{DLCT}_F(a,b) =|\{x \in \F_2^n\, \colon\, b \cdot F(x) = b  \cdot F(x+a)\}| - 2^{n-1},
\end{equation*}
yielding the \textit{differential-linear uniformity} of $F$:
\begin{equation*}
	\text{DL}_F=\max\limits_{a,b \in \F_2^n \setminus \{\mathbf{0}\}}\text{DLCT}_F(a,b).
\end{equation*}

\subsection{Comprehensive security evaluation}\label{subsec:security_analysis}
In this subsection, we present a comparative analysis of $\chi_{n,m}$, $\chi_n$, and $\cchi_n$ based on {the aforementioned metrics} for security analysis. In Tables~\ref{tb_bs}-\ref{tb_ws}, we conduct a detailed comparative security analysis among $\chi_{n,m}$ (including the mapping $\chi_5$, i.e., $\chi_{5,2}$),  $\cchi_n$ and its variants for $n=5, 6, 8$. 

\begin{table}[htbp]
	\centering
	\caption{The boomerang spectra of $\chi_n$, $\chi_{n,m}$, and $\rcchi_n$}\label{tb_bs}
	\begin{tabular}{l l l l}
		\toprule
		$n$ & Function & $\text{B}_F$ & \ Boomerang spectrum \\
		\midrule
		5 & $\chi_{5,3}$ & 24 & $\{0^{380},2^{80},4^{210},6^{40},8^{155},10^{35},14^{15},16^{30},18^5,22,24^{10}\}$ \\
		5 & $\chi_5$ & 16 & $\{0^{445},2^{176},4^{150},8^{110},12^{50},16^{30}\}$  \\
		\midrule
		6 & $\chi_{6,4}$ & 58 & $\{0^{36},4^{756},6^{24},8^{1122},10^{48},12^{372},14^{36},16^{606},18^{36},$\\
		& & & $\ 20^{246},22^{60},24^{210},26^{48},28^{204},30^9,32^{12},34^{24},36^{18},$  \\
		& & & $\ 38^3,40^{24},42^{12},44^{12},46^{13},48^{15},50^{24},54^6,58^2\}$  \\
		6 & $\chi_3 \parallel \chi_3$ & 64 & $\{0^{2247},4^{784},16^{840},64^{98}\}$  \\
		\midrule
		8 & $\chi_{8,3}$ & 224 & $\{0^{29228},2^{1224},4^{5550},6^{344},8^{5848},10^{360},12^{2068},14^{224},$ \\
		& & &  $ \ 16^{4772},18^{278},20^{1552},22^{128},24^{1784},26^{252},28^{648},  $ \\
		& & &  $ \ 30^{112},32^{2944},34^{168},36^{369},38^{56},40^{800},42^{152},44^{304},  $ \\
		& & &  $ \ 46^{24},48^{1184},50^{48},52^{224},54^4,56^{344},58^{64},60^{160}, $ \\
		& & & $ \
		64^{1384},66^{24},68^{64},70^8,72^{240},74^{48},76^{80},80^{496},$ \\
		& & & $ \
		82^{16},84^{28},88^{240},90^{24},92^{44},96^{160},100^{24},104^{64},$ \\
		& & & $ \ 108^8,112^{200},114^8,116^{16},118^8,120^{40},128^{112},$ \\
		& & & $\ 136^{112},140^{24},142^8,144^{56},146^8,152^{64}, 156^{16},$\\
		& & & $\ 160^{72},176^{32},184^8,192^{40},200^{24},224^8\}$\\
		8 & $\cchi_8$ & 256 & $\{0^{40639},4^{4928},8^{4200},16^{5720},24^{1400},32^{3090},$ \\
		& & & $\ 64^{3414},96^{750},128^{450},256^{434}\}$  \\
		\bottomrule
	\end{tabular}
\end{table}


\begin{table}[htbp]
	\centering
	\caption{The DLCT of $\chi_n$, $\chi_{n,m}$, and $\rcchi_n$}\label{tb_dlct}
	\begin{tabular}{l l l l}
		\toprule
		$n$ & Function & $\text{DL}_F$ & \ Differential-linear connectivity table \\
		\midrule
		5 & $\chi_{5,3}$ & 16 & $\{0^{285},-4^{170},4^{166},-8^{180},8^{140},-16^{15},16^{36}\}$ \\
		5 & $\chi_5$ & 16 & $\{0^{870},-16^{61},16^{61}\}$  \\
		\midrule
		6 & $\chi_{6,4}$ & 32 & $\{0^{246},-4^{223},4^{233},-8^{386},8^{434},-12^{300},12^{279},-16^{522},$ \\
		& & & $\ 16^{519},-20^{146},20^{123},-24^{294},24^{230},-32^{18},32^{69}\}$  \\
		6 & $\chi_3 \parallel \chi_3$ & 32 & $\{0^{3612},-32^{210},32^{210}\}$  \\
		\midrule
		8 & $\chi_{8,3}$ & 128 & $\{0^{26352},-16^{7224},16^{7230},-32^{8584},32^{8733},$\\
		& & & $\ -64^{3298},64^{2960},-128^{384},128^{515}\}$ \\
		8 & $\cchi_8$ & 128 & $\{0^{62148},-128^{1566},128^{1566}\}$   \\
		\bottomrule
	\end{tabular}
\end{table}

In order to investigate the algebraic degree of $\cchi_n^{-1}$, we prove the conjecture proposed in \cite{belkheyar2025chi} as a by-product of our study. proof of the following theorem is given in Appendix \ref{appendix}.

\begin{theorem} \label{th-cchi}
Let $k$ be an even integer. Every non-trivial component of the inverse of $\rcchi_{2k}$ has algebraic degree $k$.
\end{theorem}

\begin{table}[htbp]
	\centering
	\caption{The algebraic degree and the differential spectra of $\chi_n$, $\chi_{n,m}$, and $\rcchi_n$}\label{tb_ds}
	\begin{tabular}{l l l l l}
		\toprule
		$n$ & Function & Degree & $\Delta_F$ & \ Differential spectrum \\
		\midrule 
		5 & $\chi_{5,3}$ & 3 & 14 & $\{0^{721},2^{126},4^{90},6^{45},8^5,14^5\}$ \\
		5 & $\chi_5$ & 2 & 8 & $\{0^{676},2^{176},4^{120},8^{20}\}$  \\
		\midrule
		6 & $\chi_{6,4}$ & 4 & 38 & $\{0^{3441},2^{168},4^{144},6^{120},8^{96},22^{21},24^{15},26^{12},30^9,38^6\}$  \\
		6 & $\chi_3 \parallel \chi_3$ & 2 & 16 & $\{0^{3192},4^{784},16^{56}\}$  \\
		\midrule
		8 & $\chi_{8,3}$ & 3 & 112 & $\{0^{56681},2^{1896},4^{2045},6^{980},8^{1544},10^{352},12^{720},14^{176},$ \\
		& & & & $ \  
		16^{360},18^{106},20^{112},22^{40},24^{64},26^8,28^{48},30^{16}, $ \\
		& & & & $ \ 32^{24},34^8,36^{24},38^8,40^{16},44^{12},48^{16},60^8,112^8\}$ \\
		8 & $\cchi_8$ & 2 & 64 & $\{0^{56088},4^{4928},8^{3360},16^{736},32^{120},64^{48}\}$   \\
		\bottomrule
	\end{tabular}
\end{table}

 \begin{table}[htbp]
	\centering
	\caption{The nonlinearities and Walsh spectra of $\chi_n$, $\chi_{n,m}$, and $\rcchi_n$}\label{tb_ws}
	\begin{tabular}{l l l l}
		\toprule
		$n$ & Function & $\text{NL}_F$ & \ Walsh spectrum \\
		\midrule
		5 & $\chi_{5,3}$ & 4 & $\{0^{657},-8^{101},8^{230},-16^{20},16^{10},24^5,32\}$ \\
		5 & $\chi_5$ &  8 & $\{0^{647},-8^{126},8^{210},-16^{10},16^{30},32\}$  \\
		\midrule
		6 & $\chi_{6,4}$ & 4 & $\{0^{2400},-8^{480},8^{936},-16^{144},16^{49},-24^{24},24^{6},32^{15},40^{20},48^{15},64\}$  \\
		6 & $\chi_3 \parallel \chi_3$ & 16 & $\{0^{3255},-16^{490},16^{294},-32^{14},32^{42},64\}$  \\
		\midrule
		8 & $\chi_{8,3}$ & 32 & $\{0^{42040},-16^{7332},16^{10464},-32^{1712},32^{1650},-48^{520},$ \\
		& & & $\ 48^{968},-64^{257},64^{300},-80^{96},80^{40},-96^{48},96^{24},$ \\
		& & & $\ -112^{32},-128^{16},128^{16},144^{12},192^8,256\}$ \\
		8 & $\cchi_8$ & 64 & $\{0^{54603},-32^{4116},32^{5292},-64^{546},64^{910},-128^{17},128^{51},256\}$   \\
		\bottomrule
	\end{tabular}
\end{table}

Let $n=2k$ with even $k>0$. Due to the above result, $\deg(\cchi_{n}^{-1}) = k$ while {$\deg(\chi_{n,m}^{-1})=(m-1)\left \lfloor \frac{n}{m} \right \rfloor+1 > k$} by Corollary \ref{cor-chinm}. Meanwhile, $\deg(\chi_{n,m}) > \deg(\cchi_{n})=2$ for $n > m \geq 3$ with $m \nmid n$. 
Notably, $\chi_{6,4}$ and $\chi_{8,3}$ exhibit stronger resistance against boomerang ( see Table \ref{tb_bs}) and differential-linear ( see Table \ref{tb_dlct}) attacks compared to their counterparts $\chi_3 \parallel \chi_3$ and $\cchi_8$, respectively. Moreover, although $\chi_{5,3}$, $\chi_{6,4}$, and $\chi_{8,3}$ have the same differential-linear uniformity with their counterparts, our newly constructed functions exhibit significantly reduced peak frequencies in their DLCTs. Therefore, from the perspective of differential-linear cryptanalysis, the proposed mappings are slightly superior to their counterparts. 

However, our presented mappings exhibit worse performance both in differential uniformity ( see Table \ref{tb_ds}) and nonlinearity ( see Table \ref{tb_ws}) for $n=5,6,8$ compared to their counterparts.

\subsection{Implementation cost}

Among our generalized $\chi_n$-type mappings, $\chi_{n,3}$ emerges as one of the most lightweight candidates, without considering the original function. To optimize hardware implementation, we propose a modified variant $\chi'_{n,3} $ defined by $\chi'_{n,3}(x)=\chi_{n,3}(x+\1)$, which can be equivalently expressed as:
\begin{equation*}
	y_i = x_i +x_{i+1}x_{i+2}(x_{i+3}+1)  
\end{equation*}
where $i \in [n]$ and the indices are computed modulo $n$. Crucially, $\chi'_{n,3}$ preserves all essential cryptographic properties of $\chi_{n,3}$, including differential uniformity, nonlinearity, boomerang uniformity, and differential-linear uniformity, as verified through formal analysis.

We conduct a comparative study of implementation costs among  $\chi'_{n,3}$, $\chi_n$ (i.e.,  $\chi_{n,2}$) and $\cchi_n$ for both odd and even $n$. Our evaluation using standard Complementary Metal Oxide Semiconductor (CMOS) technology libraries reveals comparable implementation costs across these mappings (see Appendix: Table~\ref{tb_gate}). Figure~\ref{fig:chi} and Figure~\ref{fig:chi53} illustrate circuit implementations for $\chi_{5}$ and $\chi'_{5,3}$, respectively, revealing their structural similarities.

\begin{figure}[htbp]
	\centering
	\begin{minipage}[b]{0.48\textwidth}
		\centering
		\begin{circuitikz}[scale=0.8]
			\ctikzset{logic ports=ieee,
				logic ports/scale=0.4,        
			}
			
			\node[left] (x0) at (0,0) {$x_0$};
			\node[left] (x1) at (0,-1) {$x_1$};
			\node[left] (x2) at (0,-2) {$x_2$};
			\node[left] (x3) at (0,-3) {$x_3$};
			\node[left] (x4) at (0,-4) {$x_4$};
			
			\node[left] (y0) at (7,-0.16) {$y_0$};
			\node[left] (y1) at (7,-1.16) {$y_1$};
			\node[left] (y2) at (7,-2.16) {$y_2$};
			\node[left] (y3) at (7,-3.16) {$y_3$};
			\node[left] (y4) at (7,-4.16) {$y_4$};
			\node[and  port] (nand3-1) at (4.5,-0.5) {};
			\node[and  port] (nand3-2) at (4.5,-1.5) {};
			\node[and  port] (nand3-3) at (4.5,-2.5) {};
			\node[and  port] (nand3-4) at (4.5,-3.5) {};
			\node[and  port] (nand3-5) at (4.5,-4.5) {};
			
			\node[not port, scale=0.5] (not-1) at (3.75,-0.64) {};
			\node[not port, scale=0.5] (not-2) at (3.75,-1.64) {};
			\node[not port, scale=0.5] (not-3) at (3.75,-2.64) {};
			\node[not port, scale=0.5] (not-4) at (3.75,-3.64) {};
			\node[not port, scale=0.5] (not-5) at (3.75,-4.64) {};		
			
			\node[xor port] (xor-1) at (5.8,-0.14) {};
			\node[xor port] (xor-2) at (5.8,-1.14) {};
			\node[xor port] (xor-3) at (5.8,-2.14) {};
			\node[xor port] (xor-4) at (5.8,-3.14) {};
			\node[xor port] (xor-5) at (5.8,-4.14) {};
			
			
			\draw (x2) --++(0.9,0)  |- (nand3-1.in 1);
			\draw (x1) --++(1.2,0) |- (not-1);
			
			\draw (x3) --++(1.5,0) |- (nand3-2.in 1);
			\draw (x2) --++(1.8,0) |- (not-2);
			
			\draw (x4) --++(2.1,0) |- (nand3-3.in 1);
			\draw (x3) --++(2.4,0) |- (not-3);
			
			\draw (x0) --++(2.7,0) |- (nand3-4.in 1);
			\draw (x4) --++(3,0) |- (not-4);
			
			\draw (x1) --++(3.3,0) |- (nand3-5.in 1);
			\draw (x0) --++(3.6,0) |- (not-5);
			
			\draw (x0) -- ++(3.75,0) |- (xor-1.in 1);
			\draw (x1) -- ++(4,0) |- (xor-2.in 1);    
			\draw (x2) -- ++(4.25,0) |- (xor-3.in 1); 
			\draw (x3) -- ++(4.5,0) |- (xor-4.in 1);  
			\draw (x4) -- ++(4.5,0) |- (xor-5.in 1); 
			
			\draw (nand3-1.out) -- ++(0.1,0) |- (xor-1.in 2); 
			\draw (nand3-2.out) -- ++(0.1,0) |- (xor-2.in 2); 
			\draw (nand3-3.out) -- ++(0.1,0) |- (xor-3.in 2); 
			\draw (nand3-4.out) -- ++(0.1,0) |- (xor-4.in 2); 
			\draw (nand3-5.out) -- ++(0.1,0) |- (xor-5.in 2);
		\end{circuitikz}
		\caption{A circuit implementation of $\chi_{5}$}\label{fig:chi}
	\end{minipage}
	\hfill 
\begin{minipage}[b]{0.48\textwidth}
		\centering
		\begin{circuitikz}[scale=0.8]
			\ctikzset{logic ports=ieee,
				logic ports/scale=0.4,        
			}
			
			\node[left] (x0) at (0,0) {$x_0$};
			\node[left] (x1) at (0,-1) {$x_1$};
			\node[left] (x2) at (0,-2) {$x_2$};
			\node[left] (x3) at (0,-3) {$x_3$};
			\node[left] (x4) at (0,-4) {$x_4$};
			
			\node[left] (y0) at (7.5,-0.16) {$y_0$};
			\node[left] (y1) at (7.5,-1.16) {$y_1$};
			\node[left] (y2) at (7.5,-2.16) {$y_2$};
			\node[left] (y3) at (7.5,-3.16) {$y_3$};
			\node[left] (y4) at (7.5,-4.16) {$y_4$};
			\node[nand  port, number inputs=3] (nand3-1) at (5,-0.5) {};
			\node[nand  port, number inputs=3] (nand3-2) at (5,-1.5) {};
			\node[nand  port, number inputs=3] (nand3-3) at (5,-2.5) {};
			\node[nand  port, number inputs=3] (nand3-4) at (5,-3.5) {};
			\node[nand  port, number inputs=3] (nand3-5) at (5,-4.5) {};
			
			\node[not port, scale=0.5] (not-1) at (4.25,-0.686) {};
			\node[not port, scale=0.5] (not-2) at (4.25,-1.686) {};
			\node[not port, scale=0.5] (not-3) at (4.25,-2.686) {};
			\node[not port, scale=0.5] (not-4) at (4.25,-3.686) {};
			\node[not port, scale=0.5] (not-5) at (4.25,-4.686) {};		
			
			\node[xor port] (xor-1) at (6.3,-0.14) {};
			\node[xor port] (xor-2) at (6.3,-1.14) {};
			\node[xor port] (xor-3) at (6.3,-2.14) {};
			\node[xor port] (xor-4) at (6.3,-3.14) {};
			\node[xor port] (xor-5) at (6.3,-4.14) {};
			
			
			\draw (x1) --++(0.75,0)  |- (nand3-1.in 1);
			\draw (x2) --++(1,0) |- (nand3-1.in 2);
			\draw (x3) --++(1.25,0) |- (not-1);
			
			\draw (x2) --++(1.5,0) |- (nand3-2.in 1);
			\draw (x3) --++(1.75,0) |- (nand3-2.in 2);
			\draw (x4) --++(2,0) |- (not-2);
			
			\draw (x3) --++(2.25,0) |- (nand3-3.in 1);
			\draw (x4) --++(2.5,0) |- (nand3-3.in 2);
			\draw (x0) --++(2.75,0) |- (not-3);
			
			\draw (x4) --++(3,0) |- (nand3-4.in 1);
			\draw (x0) --++(3.25,0) |- (nand3-4.in 2);
			\draw (x1) --++(3.5,0) |- (not-4);
			
			\draw (x0) --++(3.75,0) |- (nand3-5.in 1);
			\draw (x1) --++(4,0) |- (nand3-5.in 2);
			\draw (x2) --++(4.25,0) |- (not-5);
			
			\draw (x0) -- ++(3.75,0) |- (xor-1.in 1);
			\draw (x1) -- ++(4,0) |- (xor-2.in 1);    
			\draw (x2) -- ++(4.25,0) |- (xor-3.in 1); 
			\draw (x3) -- ++(4.5,0) |- (xor-4.in 1);  
			\draw (x4) -- ++(4.5,0) |- (xor-5.in 1); 
			
			\draw (nand3-1.out) -- ++(0.1,0) |- (xor-1.in 2); 
			\draw (nand3-2.out) -- ++(0.1,0) |- (xor-2.in 2); 
			\draw (nand3-3.out) -- ++(0.1,0) |- (xor-3.in 2); 
			\draw (nand3-4.out) -- ++(0.1,0) |- (xor-4.in 2); 
			\draw (nand3-5.out) -- ++(0.1,0) |- (xor-5.in 2);
		\end{circuitikz}
		\caption{A circuit implementation of $\chi'_{5,3}$}\label{fig:chi53}
	\end{minipage}
\end{figure}

The component function $y_i=x_i+x_{i+1}x_{i+2}(x_{i+3}+1)$ in $\chi'_{n,3}$ requires one XOR-gate, one NOT-gate, and one $3$-input NAND-gate per bit. In contrast, $\chi_{n,2}$ replaces the $3$-input NAND-gate with a $2$-input AND-gate. As demonstrated in Table~\ref{tb_ge}, this substitution results in similar implementation areas for both mappings. Moreover, $\chi'_{n,3}$ achieves an implementation area no larger than $\cchi_n$ by eliminating one NOT-gate.

Latency complexity represents another critical metric for lightweight implementations. Following the framework in~\cite[Sec.3]{Rsa2022}, which defines latency as the longest path through NAND/NOR-gates using the gate set $\mathcal{G} = \{\text{NOT}, \text{NAND}, \text{NOR}\}$, we analyze critical paths:
\begin{itemize}
\item $\chi'_{n,3}$ exhibits a 4-stage latency (2 stages for 3-input NAND + 2 stages for XOR);
\item Both $\chi_n$ and $\cchi_n $ show 3-stage latency (1 stage for AND + 2 stages for XOR).
\end{itemize}
This analysis suggests that while $\chi'_{n,3}$ maintains area efficiency, it incurs a modest latency tradeoff compared to $\chi_n$-type mappings. The choice between these implementations may therefore depend on specific application constraints emphasizing either area or latency optimization.


\section{Conclusion and future work} \label{sec:conclusion}

In this work, we generalized the well-known $\chi_n$ permutation to an extensive family of permutations $\chi_{n,m}$, which exists both for odd and even dimension vector spaces over $\F_2$. 
We not only provide a clear characterization of the algebraic structure of $\chi_{n,m}$ but also conduct a comprehensive comparative analysis between our construction and their counterparts for small-scale $n$, including implementation costs and security metrics. Our research establishes a solid foundation for subsequent investigations in both theoretical and applied domains. Two clear paths to continue are suggested: $\chi_n$, $\chi_{n,3}$ and $\cchi_n$ show advantages for diverse security metrics, and it would be of interest to investigate their concatenations to trade-off these metrics, {\it e.g.}, $\chi_3 \parallel \chi_{5,3} \parallel \cchi_{8}$.
The components of $\cchi_{n}$ and $\chi_n$ are similar and have the same algebraic degree, and so another open security path is to study the algebraic structure of $\cchi_n$ and their iterates, extending our methods..


\newpage
\appendix

\section{proof of Theorem \ref{th-cchi}} \label{appendix}

\begin{proof}
	Based on (\ref{eq-cchi}), we first attempt to express $x_i$ in terms of $x_{k-3},x_{k-2},x_{k-1},x_k$, and $y_0,y_1,\ldots,y_{2k-1}$.
	Since $y_i = x_i + \overline{x}_{i+1} x_{i+2}$ for $i < k-3$ or $k < i < n-2$, $x_{k-4}$ can be expressed as:
	\begin{equation*}
		x_{k-4} = y_{k-4}+\bar{x}_{k-3}x_{k-2}.
	\end{equation*}
	From $x_{k-4}$ we can further derive an alternative expression for $x_{k-5}$:
	\begin{equation*}
		x_{k-5} = y_{k-5}+\bar{x}_{k-4}x_{k-3}
		= y_{k-5}+\bar{y}_{k-4}x_{k-3}.
	\end{equation*}
	Therefore, we can recursively express $x_i$ as follows for $0 \le i \le k-4$:
	\begin{equation} \label{eq-0lem-4}
		x_i=\left\{\begin{aligned}
			y_{i}+\bar{y}_{i+1}\Big(y_{i+2}+\bar{y}_{i+3}\big(\cdots\bar{y}_{k-5}\left(y_{k-4}+\bar{x}_{k-3}x_{k-2}\right)\cdots\big)\Big), & \quad {\rm if} \ i \ {\rm is \ even},\\ y_{i}+\bar{y}_{i+1}\Big(y_{i+2}+\bar{y}_{i+3}\big(\cdots\bar{y}_{k-6}(y_{k-5}+\bar{y}_{k-4}x_{k-3})\cdots\big)\Big), & \quad {\rm if} \ i \ {\rm is \ odd}.
		\end{aligned} \right.
	\end{equation}
	Note that the first two terms are
	\begin{equation} \label{eq-x0x1}
		\left\{\begin{aligned}
			x_1&=y_1+\bar{y}_2\Big(y_3+\bar{y}_4\big(\cdots \bar{y}_{k-6}(y_{k-5}+\bar{y}_{k-4}x_{k-3})\cdots\big)\Big),\\
			x_0&=y_0+\bar{y}_1\Big(y_2+\bar{y}_3\big(\cdots \bar{y}_{k-5}(y_{k-4}+\bar{x}_{k-3}x_{k-2})\cdots\big)\Big).\\
		\end{aligned} \right.
	\end{equation}

Similarly, for $k+1 \le i \le 2k-1$, $x_i$ can be recursively expressed as follows:
\begin{equation} \label{eq-m+1le2m-1}
	\left\{\begin{aligned} 
		x_{2k-1} &=y_{2k-1}+\bar{x}_{k-1}x_k,\\
		x_{2k-2}&=y_{2k-2}+\bar{x}_{2k-1}x_{k-1} =y_{2k-2}+\bar{y}_{2k-1}x_{k-1},\\
		\vdots &\\
		x_{k+2t+1}&=y_{k+2t+1}+\bar{y}_{k+2t+2}\big(\cdots \bar{y}_{2k-2}(y_{2k-1}+\bar{x}_{k-1}x_k)\cdots\big), \\
		x_{k+2t}&=y_{k+2t}+\bar{y}_{k+2t+1}\big(\cdots \bar{y}_{2k-3}(y_{2k-2}+\bar{y}_{2k-1}x_{k-1})\cdots\big),\\
		\vdots &\\
		x_{k+2}&=y_{k+2}+\bar{y}_{k+3}\big(\cdots \bar{y}_{2k-3}(y_{2k-2}+\bar{y}_{2k-1}x_{k-1})\cdots\big),\\
		x_{k+1}&=y_{k+1}+\bar{y}_{k+2}\big(\cdots \bar{y}_{2k-2}(y_{2k-1}+\bar{x}_{k-1}x_k)\cdots\big).
	\end{aligned}  \right.
\end{equation}
The remaining four terms could not be simply represented recursively and are as follows:
\[
\left\{\begin{aligned}
	x_{k-3}&=\bar{y}_{k-1}+\bar{x}_k\bar{x}_{k+1}, \\
	x_{k-2}&=y_{k}+\bar{x}_{k+1}x_{k+2}=y_{k}+\bar{y}_{k+1}x_{k+2}, \\
	x_{k-1}&=y_{k-2}+\bar{x}_0x_1=y_{k-2}+\bar{y}_0x_1, \\
	x_k &=y_{k-3}+\bar{x}_{k-2}x_0.
\end{aligned} \right.
\]
By substituting the expressions of $x_{k+1}$, $x_{k+2}$ in (\ref{eq-m+1le2m-1}) and $x_{1}$, $x_{0}$ in (\ref{eq-x0x1}) into the expressions for $x_{k-3}$, $x_{k-2}$, $x_{k-1}$, and $x_{k}$ respectively, we obtain:
\begin{equation} \label{eq-3210}
	\left\{\begin{aligned}
		x_{k-3}
		&=\bar{y}_{k-1}+\bar{x}_k\Big(\bar{y}_{k+1}+\bar{y}_{k+2}\big(\cdots \bar{y}_{2k-2}(y_{2k-1}+\bar{x}_{k-1}x_k)\cdots\big)\Big)\\
		&=\bar{y}_{k-1}+\bar{x}_k\Big(\bar{y}_{k+1}+\bar{y}_{k+2}\big(\cdots(y_{2k-3}+\bar{y}_{2k-2}y_{2k-1})\cdots\big)\Big),\\
		x_{k-2}	&=y_{k}+\bar{y}_{k+1}\Big(y_{k+2}+\bar{y}_{k+3}\big(\cdots \bar{y}_{2k-3}(y_{2k-2}+\bar{y}_{2k-1}x_{k-1})\cdots\big)\Big)\\
		&=y_{k}+\bar{y}_{k+1}\bar{y}_{k+3}\cdots \bar{y}_{2k-1} x_{k-1} \\
		&\quad+\bar{y}_{k+1} \Big(y_{k+2}+\bar{y}_{k+3}\big(\cdots (y_{2k-4}+\bar{y}_{2k-3}y_{2k-2})\cdots\big)\Big),\\
		x_{k-1}
		&=y_{k-2}+\bar{y}_0\Big(y_1+\bar{y}_2\big(\cdots \bar{y}_{k-6}(y_{k-5}+\bar{y}_{k-4}x_{k-3})\cdots\big)\Big)\\
		&=y_{k-2}+\bar{y}_0\bar{y}_2\cdots\bar{y}_{k-4} x_{k-3}+\bar{y}_0\Big(y_1+\bar{y}_2\big(\cdots (y_{k-7}+\bar{y}_{k-6}y_{k-5})\cdots\big)\Big),\\
		x_k 
		&=y_{k-3}+\bar{x}_{k-2}\Big(y_0+\bar{y}_1\big(\cdots \bar{y}_{k-5}(y_{k-4}+\bar{x}_{k-3}x_{k-2})\cdots\big)\Big)\\
		&=y_{k-3}+\bar{x}_{k-2} \Big(y_0+\bar{y}_1\big(\cdots (y_{k-6}+\bar{y}_{k-5}y_{k-4})\cdots\big)\Big).
	\end{aligned} \right.
\end{equation}

Furthermore, by substituting the expressions of $x_{k}$ and $x_{k-1}$ into the expressions for $x_{k-3}$ and $x_{k-2}$, respectively, we have
\begin{equation} \label{eq-xm-3m-2}
	\left\{\begin{aligned}
		x_{k-3}&=\bar{y}_{k-1}+\bigg(\bar{y}_{k-3}+\bar{x}_{k-2}\Big(y_0+\bar{y}_1\big(\cdots (y_{k-6}+\bar{y}_{k-5}y_{k-4})\cdots\big)\Big)\bigg) \\
		&\quad \cdot\Big(\bar{y}_{k+1}+\bar{y}_{k+2}\big(\cdots(y_{2k-3}+\bar{y}_{2k-2}y_{2k-1})\cdots\big)\Big)\\
		&=\bar{y}_{k-1}+\bar{x}_{k-2}\cdot f\cdot g+\bar{y}_{k-3} \cdot g ,\\
		x_{k-2}&=y_{k}+\bar{y}_{k+1}\bar{y}_{k+3}\cdots \bar{y}_{2k-1} \bigg(y_{k-2}+\bar{y}_0\bar{y}_2\cdots\bar{y}_{k-4} x_{k-3} \\
		&\quad +\bar{y}_0\Big(y_1+\bar{y}_2\big(\cdots (y_{k-7}+\bar{y}_{k-6}y_{k-5})\cdots\big)\Big)\bigg)  \\
		&\quad +\bar{y}_{k+1}\Big(y_{k+2}+\bar{y}_{k+3}\big(\cdots (y_{2k-4}+\bar{y}_{2k-3}y_{2k-2})\cdots\big)\Big) \\
		&=y_k+h\cdot \bigg(y_{k-2}+\bar{y}_0\Big(y_1+\bar{y}_2\big(\cdots (y_{k-7}+\bar{y}_{k-6}y_{k-5})\cdots\big)\Big)\bigg)\\
		&\quad+x_{k-3}\cdot e\cdot h+\bar{y}_{k+1} \Big(y_{k+2}+\bar{y}_{k+3}\big(\cdots (y_{2k-4}+\bar{y}_{2k-3}y_{2k-2})\cdots\big)\Big),
	\end{aligned} \right.
\end{equation}
where
\[
\left\{\begin{aligned}
	f&=y_0+\bar{y}_1\Big(y_2+\bar{y}_3\big(\cdots (y_{k-6}+\bar{y}_{k-5}y_{k-4})\cdots\big)\Big),\\
	g&=\bar{y}_{k+1}+\bar{y}_{k+2}\Big(y_{k+3}+\bar{y}_{k+4}\big(\cdots(y_{2k-3}+\bar{y}_{2k-2}y_{2k-1})\cdots\big)\Big),\\
	e&=\bar{y}_0\bar{y}_2\cdots\bar{y}_{k-4},\\
	h&=\bar{y}_{k+1}\bar{y}_{k+3}\cdots \bar{y}_{2k-1}.
\end{aligned} \right.
\]
The degrees of $f$, $g$, $e$, and $h$ are $\frac{k}{2}-1$, $\frac{k}{2}$, $\frac{k}{2}-1$, and $\frac{k}{2}$ respectively. 
Some interesting properties include $g \cdot h = h$ and $f \cdot e = 0$.

Substituting $x_{k-2}$ into the expression of $x_{k-3}$, or $x_{k-3}$ into the expression of $x_{k-2}$ would lead to a multiplication of $f\cdot g\cdot e\cdot h$ which equals $0$ since $f\cdot e = 0$, we know that the expressions of $x_{k-3}$ 
and $x_{k-2}$ involve only $y_i$. 

The second term of $x_{k-3}$ in (\ref{eq-xm-3m-2})
\begin{equation} \label{eq-xm-2fg}
	\begin{aligned}
		&\bar{x}_{k-2}\cdot f\cdot g \\
		&=\bar{y}_k\cdot f\cdot g + h\cdot\bigg(y_{k-2}+\bar{y}_0\Big(y_1+\bar{y}_2\big(\cdots (y_{k-7}+\bar{y}_{k-6}y_{k-5})\cdots\big)\Big)\bigg)\cdot f\cdot g \\
		&\quad+\bar{y}_{k+1} \Big(y_{k+2}+\bar{y}_{k+3}\big(\cdots (y_{2k-4}+\bar{y}_{2k-3}y_{2k-2})\cdots\big)\Big)\cdot f\cdot g
	\end{aligned}
\end{equation}
has degree $k$, since the second term of $\bar{x}_{k-2}\cdot f\cdot g$ in (\ref{eq-xm-2fg})
\[
\begin{aligned}
	&h\cdot\bigg(y_{k-2}+\bar{y}_0 \Big(y_1+\bar{y}_2\big(\cdots (y_{k-7}+\bar{y}_{k-6}y_{k-5})\cdots\big)\Big)\bigg)\cdot f\cdot g\\
	&=h\cdot g \cdot \Bigg(\bigg(y_{k-2}+\bar{y}_0 \Big(y_1+\bar{y}_2\big(\cdots (y_{k-7}+\bar{y}_{k-6}y_{k-5})\cdots\big)\Big)\bigg)\cdot f\Bigg)\\
	&=\bar{y}_{k+1}\bar{y}_{k+3}\cdots \bar{y}_{2k-1}y_{k-2}\cdot f
\end{aligned}
\]
has degree $k$, and the last term of $\bar{x}_{k-2}\cdot f\cdot g$ in (\ref{eq-xm-2fg})
\[
\begin{aligned}
	&\bar{y}_{k+1} \Big(y_{k+2}+\bar{y}_{k+3}\big(\cdots (y_{2k-4}+\bar{y}_{2k-3}y_{2k-2})\cdots\big)\Big)\cdot f\cdot g\\
	&=\bar{y}_{k+1}  \Big(y_{k+2}+\bar{y}_{k+3}\big(\cdots (y_{2k-4}+\bar{y}_{2k-3}y_{2k-2})\cdots\big)\Big)\cdot f
\end{aligned}
\]
has degree of $k-1$. Combining this with the fact that the last term of $x_{k-3}$ in (\ref{eq-xm-3m-2}) is $\bar{y}_{k-3} \cdot g$, which has degree of $\frac{k}{2}+1$. So the algebraic degree $x_{k-3}$ is $k$. Similarly, we can obtain that the degree of $x_{k-2}$ is $k$.

Rewriting the expression for $x_{k-3}$ yields
\begin{equation} \label{eq-xm-3}
	\begin{aligned}
		x_{k-3} &= \bar{y}_{k-1}+\bar{y}_{k-3}\Big(\bar{y}_{k+1}+\bar{y}_{k+2}\big(\cdots(y_{2k-3}+\bar{y}_{2k-2}y_{2k-1})\cdots\big)\Big) \\
		&\quad +f\cdot y_{k-2}\bar{y}_{k+1}\bar{y}_{k+3}\cdots\bar{y}_{2k-1}+ f\cdot \bar{y}_{k+1} \Big(y_{k+2} \\
		&\quad +\bar{y}_{k+3}\big(\cdots (y_{2k-4}+\bar{y}_{2k-3}y_{2k-2})\cdots\big)\Big)+f\cdot \bar{y}_{k}\cdot g \\
		&=  \bar{y}_{k-1}+\bar{y}_{k-3}\cdot g+f\cdot h\cdot y_{k-2}+f\cdot F+f\cdot \bar{y}_{k}\cdot g,\\
	\end{aligned}
\end{equation}
where 
\[
F=\bar{y}_{k+1} \Big(y_{k+2}+\bar{y}_{k+3}\big(\cdots (y_{2k-4}+\bar{y}_{2k-3}y_{2k-2})\cdots\big)\Big),
\]
has degree $\frac{k}{2}$ and $g\cdot F = F$.

Back to $x_0,\ldots, x_{k-4}$, we need to prove that $\bar{y}_{2i+1}\bar{y}_{2i+3}\cdots\bar{y}_{k-5}\bar{x}_{k-3}x_{k-2}$ has degree of $k$, where $i\in \{0, 1, \cdots, \frac{k}{2}-2 \}$, and $\bar{y}_{2j}\bar{y}_{2j+2}\cdots\bar{y}_{k-4} x_{k-3}$ has degree of $k$, where $j\in \{1,\ldots, \frac{k}{2}-2\}$.
Note that
\begin{equation} \label{eq-xm-3xm-2}
	\bar{x}_{k-3}x_{k-2} =\bar{x}_{k-3} (y_k+x_{k-3}\cdot e\cdot h+h\cdot G+F),
\end{equation}
where
\[
\begin{aligned}
	G=y_{k-2}+\bar{y}_0\Big(y_1+\bar{y}_2\big(\cdots (y_{k-7}+\bar{y}_{k-6}y_{k-5})\cdots\big)\Big), 
\end{aligned}
\] 
has degree $\frac{k}{2}-1$. Then, substituting $x_{k-3}$ in (\ref{eq-xm-3}) into (\ref{eq-xm-3xm-2}),  we have
\[
\begin{aligned}
	&\bar{x}_{k-3}x_{k-2} =(y_{k-1}+\bar{y}_{k-3}\cdot g+f\cdot h\cdot y_{k-2}+f\cdot F+f\cdot \bar{y}_{k}\cdot g)  (y_k+h\cdot G+F)\\
	&=y_{k-1}y_k+\bar{y}_{k-3}y_k\cdot g+y_{k-2}y_k\cdot f\cdot h+y_k\cdot f\cdot F+y_{k-1}\cdot h\cdot G+\bar{y}_{k-3}\cdot h\cdot G \\
	&\quad+ y_{k-2}\cdot f\cdot h +y_{k-2}\bar{y}_k\cdot f\cdot h+y_{k-1}\cdot F+\bar{y}_{k-3}\cdot F +f\cdot F+\bar{y}_{k}\cdot f\cdot F\\
	&=y_{k-1}y_k+\bar{y}_{k-3}y_k\cdot g+(y_{k-1}+\bar{y}_{k-3})\cdot F+(y_{k-1}+\bar{y}_{k-3})\cdot h\cdot G.
\end{aligned}
\]
Next, we consider the last term of $\bar{x}_{k-3}x_{k-2}$. Note that 
\[
\begin{aligned}
	&\bar{y}_{2t+1} \bar{y}_{2t+3} \cdots \bar{y}_{k-5}\cdot G \\
	&= \bar{y}_{2t+1} \bar{y}_{2t+3} \cdots \bar{y}_{k-5} \bigg(y_{k-2}+\bar{y}_0 \Big(y_1+\bar{y}_2\big(\cdots(y_{2t-3}+\bar{y}_{2t-2}y_{2t-1})\cdots \big)\Big)\bigg)
\end{aligned}
\]
always has degree $\frac{k}{2}-1$ for $t \in [\frac{k}{2}-2]$. Therefore, combining with the fact that the degree of $h$ is $\frac{k}{2}$, we know that the degrees of $x_0, x_2,\ldots, x_{k-4}$ are all $k$ by (\ref{eq-0lem-4}). Similarly, 
\[\bar{y}_{2t} \bar{y}_{2t+2} \bar{y}_{k-4} \cdot f 
=\bar{y}_{2t} \bar{y}_{2t+2}\cdots \bar{y}_{k-4} \Big(y_0+\bar{y}_1\big(\cdots(y_{2t-4}+\bar{y}_{2t-3}y_{2t-2})\cdots\big)\Big)
\]
always has degree $\frac{k}{2}-1$ for $t\in \{1,2, \cdots, \frac{k}{2}-2\}$.
Thus, the degrees of $x_0, x_1,\ldots, x_{k-4}$ are all $k$.

To investigate the algebraic degrees of $x_{k+1}, x_{k+2},\ldots, x_{2k-1}$, we need to consider $x_{k-1}$ and $x_k$. By (\ref{eq-3210}) and (\ref{eq-xm-3}),
\[
\begin{aligned}
	x_{k-1} &= y_{k-2}+(\bar{y}_{k-1}+\bar{y}_{k-3}\cdot g+f\cdot h\cdot y_{k-2}+f\cdot F+f\cdot \bar{y}_{k}\cdot g)\cdot e+(G+y_{k-2}) \\
	&=G+\bar{y}_{k-1}\cdot e+\bar{y}_{k-3}\cdot g\cdot e+(h\cdot y_{k-2}+F+\bar{y}_k\cdot g)\cdot f\cdot e\\
	&=G+\bar{y}_{k-1}\cdot e+\bar{y}_{k-3}\cdot g\cdot e
\end{aligned}
\]
has degree $k$ since the degrees of $g$ and $e$ are $\frac{k}{2}$ and $\frac{k}{2}-1$, respectively.
Then 
\[
\begin{aligned}
	&\bar{y}_{2t+1}\bar{y}_{2t+3}\cdots \bar{y}_{2k-1}\cdot g  \\
	&= \bar{y}_{2t+1}\bar{y}_{2t+3}\cdots \bar{y}_{2k-1}  \bigg(\bar{y}_{k+1}+\bar{y}_{k+2}\Big(y_{k+3} +\bar{y}_{k+4}\big(\cdots(y_{2t-3}+\bar{y}_{2t-2}y_{2t-1})\cdots\big)\Big)\bigg)
\end{aligned}
\]
always has degree $\frac{k}{2}$ for $t\in \{\frac{k}{2}+1,\frac{k}{2}+2,\ldots, k-1\}$.
This implies that $\bar{y}_{2t+1}\bar{y}_{2t+3}\cdots \bar{y}_{2k-1}$ 
$\cdot x_{k-1}$ always has degree $k$ for $t\in \{\frac{k}{2}+1,\frac{k}{2}+2,\ldots, k-1\}$.
So $x_{k+2}, x_{k+4},\ldots,x_{2k-2}$ have degree $k$ by (\ref{eq-m+1le2m-1}). Similarly, 
\begin{equation*}
	x_k = y_{k-3}+\bar{x}_{k-2}\cdot f =y_{k-3}+f\cdot (\bar{y}_k+h\cdot x_{k-1}+F),
\end{equation*}
has degree $k$ since the degrees of $f$ and $h$ are $\frac{k}{2}-1$ and $\frac{k}{2}$, respectively. Then 
\[
\begin{aligned}
	\bar{x}_{k-1}x_k&=\bar{x}_{k-1}\cdot (y_{k-3}+f\cdot \bar{y}_k+f\cdot F)\\
	&=(\bar{G}+\bar{y}_{k-1}\cdot e)y_{k-3}+\bar{y}_{k-2}\bar{y}_k\cdot f+\bar{y}_{k-2}\cdot f\cdot F,
\end{aligned}
\]
and 
\[
\begin{aligned}
	&\bar{y}_{2t}\bar{y}_{2t+2}\cdots \bar{y}_{2k-2}\cdot F\\
	&=
	\bar{y}_{2t}\bar{y}_{2t+2}\cdots \bar{y}_{2k-2} \bigg(\bar{y}_{k+1}\Big(y_{k+2}+\bar{y}_{k+3}\big(\cdots(y_{2t-4}+\bar{y}_{2t-3}y_{2t-2})\cdots\big)\Big)\bigg)
\end{aligned}
\]
always has degree $\frac{k}{2}$ for $t\in \{\frac{k}{2}+2,\frac{k}{2}+3,\ldots, k-1\}$, and is $0$ when $t=\frac{k}{2}+1$.
So $x_{k+1}, x_{k+3},\ldots, x_{2k-1}$ have degree $k$.  This completes proof.

\end{proof}

\section{Frequently-used logic gates in common CMOS technology libraries}

\vspace{-0.2cm} 

\begin{table}[h]
	\centering
	\caption{Frequently-used logic gates in common CMOS technology libraries}\label{tb_gate}
	\begin{tabular}{c c c c}
		\toprule
		Gate & Domain $\to$ Range & Logic formulas & Algebraic normal form \\
		\midrule
		NOT & $\F_2 \to \F_2$ & $x_0 \to \neg x_0$ & $x_0 \to x_0+1$  \\
		AND & $\F_2^2 \to \F_2$ & $x_0,x_1 \to x_0 \land x_1$ & $x_0,x_1 \to x_0x_1$ \\
		OR & $\F_2^2 \to \F_2$ & $x_0,x_1 \to x_0 \lor x_1$ & $x_0,x_1 \to x_0x_1+x_0+x_1$ \\
		NAND & $\F_2^2 \to \F_2$ & $x_0,x_1 \to \neg(x_0 \land x_1)$ & $x_0,x_1 \to x_0x_1+1$ \\
		NOR & $\F_2^2 \to \F_2$ & $x_0,x_1 \to \neg(x_0 \lor x_1)$ & $x_0,x_1 \to x_0x_1+x_0+x_1+1$ \\
		XOR & $\F_2^2 \to \F_2$ & $x_0,x_1 \to x_0 \oplus x_1$ & $x_0,x_1 \to x_0+x_1$ \\
		XNOR & $\F_2^2 \to \F_2$ & $x_0,x_1 \to \neg(x_0 \oplus x_1)$ & $x_0,x_1 \to x_0+x_1+1$ \\
		AND3 & $\F_2^3 \to \F_2$ & $x_0,x_1,x_2 \to x_0 \land x_1 \land x_2$ & $x_0,x_1,x_2 \to x_0x_1x_2$ \\
		NAND3 & $\F_2^3 \to \F_2$ & $x_0,x_1,x_2 \to \neg(x_0 \land x_1 \land x_2)$ & $x_0,x_1,x_2 \to x_0x_1x_2+1$ \\
		XOR3 & $\F_2^3 \to \F_2$ & $x_0,x_1,x_2 \to x_0 \oplus x_1 \oplus x_2$ & $x_0,x_1,x_2 \to x_0+x_1+x_2$ \\
		\bottomrule
	\end{tabular}
\end{table}

\vspace{-0.2cm} 

\begin{table}[h]
	\centering
	\caption{Implementation area of frequently-used logic gates in common CMOS technology libraries}\label{tb_ge}
	\begin{tabular}{c c c c c c c c c c}
		\toprule
		Gate & UMC & TSMC & TSMC & SMIC & SMIC & Nangate & Nangate & STD & STM  \\
		& 180nm & 65nm & 28nm & 130nm & 65nm & 45nm & 15nm & 350nm & 65nm \\
		\midrule
		NOT & 0.67 & 0.50 & 0.67 & 0.67 & 0.75 & 0.67 & 0.75 & 0.67 & 0.50 \\
		AND & 1.33 & 1.50 & 1.33 & 1.33 & 1.50 & 1.33 & 1.50 & 1.33 & 1.50 \\
		OR & 1.33 & 1.50 & 1.33 & 1.33 & 1.50 & 1.33 & 1.50 & 1.33 & 1.50 \\
		NAND & 1.00 & 1.00 & 1.00 & 1.00 & 1.00 & 1.00 & 1.00 & 1.00 & 1.00 \\
		NOR & 1.00 & 1.00 & 1.00 & 1.00 & 1.00 & 1.00 & 1.00 & 1.00 & 1.00 \\
		XOR & 2.67 & 2.50 & 3.00 & 2.33 & 2.25 & 2.00 & 2.25 & 2.33 & 2.00 \\
		XNOR & 2.00 & 2.50 & 3.00 & 2.33 & 2.25 & 2.00 & 2.25 & 2.33 & 2.00 \\
		AND3 & 2.33 & 2.00 & 1.67 & 1.67 & 1.75 & 1.67 & 2.00 & 1.67 & 2.00 \\
		NAND3 & 1.33 & 1.50 & 1.33 & 1.33 & 1.25 & 1.33 & 1.50 & 1.33 & 1.50 \\
		XOR3 & 4.67 & 5.50 & 4.33 & 5.67 & 4.75 & N/A & N/A & 4.00 & \\
		\bottomrule
	\end{tabular}
\end{table}

\end{document}